\documentclass[envcountsame,11pt]{llncs}
\usepackage[a4paper,margin=1.0in]{geometry}
\usepackage{array,longtable}
\usepackage{amssymb,amsmath,verbatim,booktabs,bm}
\usepackage{xcolor,ntheorem,complexity}
\usepackage{graphicx,lmodern}
\usepackage{url,soul}
\usepackage[breaklinks=true]{hyperref}
\colorlet{green}{green!50!black}
\hypersetup{
	colorlinks=true,
	citecolor=green,
	linkcolor=blue,  
	urlcolor=black,
}
\usepackage{algorithm}
\usepackage{algorithmicx}
\usepackage[noend]{algpseudocode}

\newcommand{\myproblem}[1]{{\textsc{#1}}} 
\usepackage[textsize=scriptsize,textwidth=4cm,shadow]{todonotes}

\newcommand{\ildicom}[1]{\todo[color=teal!50!white]{I: #1}}

\def\bp{\textup{bp}}
\def\AA{\widetilde{A}}
\def\BB{\widetilde{B}}
\def\CC{\widetilde{C}}
\def\MM{\widetilde{M}}
\def\F'{\widetilde{F}}
\def\GG{\widetilde{G}}
\def\TT{\texttt{t}}
\def\FF{\texttt{f}}

\def\FUP{\myproblem{Popularity with Fixed Instability}}
\def\FUPT{\myproblem{Popularity with Fixed Instability and Ties}}
\def\MAXUP{\myproblem{Popularity with Bounded Instability}}
\def\MUPMU{\myproblem{Max-Utility Popular Matching with Instability Costs}}
\def\MUPMUT{\myproblem{Max-Utility Popular Matching with Instability Costs and Ties}}
\def\MUPOMU{\myproblem{Max-Utility Pareto-Optimal Matching with Instability Costs}}

\usepackage{environ}
\usepackage{mathdots}

\newtheorem{observation}[theorem]{Observation}

\newtheorem{myremark}{Remark}
\let\oldremark\myremark
\renewcommand{\myremark}{\oldremark\normalfont}

\newtheorem{myexample}{Example}
\let\oldexample\myexample
\renewcommand{\myexample}{\oldexample\normalfont}

\newcommand{\repeattheorem}[1]{%
  \begingroup
  \renewcommand{\thetheorem}{\ref{#1}}%
  \expandafter\expandafter\expandafter\theorem
  \csname reptheorem@#1\endcsname
  \endtheorem
  \endgroup
}
\NewEnviron{reptheorem}[1]{%
  \global\expandafter\xdef\csname reptheorem@#1\endcsname{%
    \unexpanded\expandafter{\BODY}%
  }%
  \expandafter\theorem\BODY\unskip\label{#1}\endtheorem
}

\newcommand{\repeatlemma}[1]{%
  \begingroup
  \renewcommand{\thelemma}{\ref{#1}}%
  \expandafter\expandafter\expandafter\lemma
  \csname replemma@#1\endcsname
  \endlemma
  \endgroup
}
\NewEnviron{replemma}[1]{%
  \global\expandafter\xdef\csname replemma@#1\endcsname{%
    \unexpanded\expandafter{\BODY}%
  }%
  \expandafter\lemma\BODY\unskip\label{#1}\endlemma
}

\newcommand{\repeatcorollary}[1]{%
  \begingroup
  \renewcommand{\thecorollary}{\ref{#1}}%
  \expandafter\expandafter\expandafter\corollary
  \csname repcorollary@#1\endcsname
  \endcorollary
  \endgroup
}
\NewEnviron{repcorollary}[1]{%
  \global\expandafter\xdef\csname repcorollary@#1\endcsname{%
    \unexpanded\expandafter{\BODY}%
  }%
  \expandafter\corollary\BODY\unskip\label{#1}\endcorollary
}

\makeatletter
\renewcommand\paragraph{\@startsection{paragraph}{4}{\z@}%
                                    {0.6ex \@plus1ex \@minus.2ex}%
                                    {-1em}%
                                   {\normalfont\normalsize\bfseries}} 
\makeatother

\title{Maximum-utility popular matchings with bounded instability
}
\author{Ildik\'{o} Schlotter\inst{1,2} \and \'{A}gnes Cseh\inst{1}}
\institute{
 Centre for Economic and Regional Studies, Hungary
\email{\{schlotter.ildiko,cseh.agnes\}@krtk.hu}
\and
Budapest University of Technology and Economics,  Hungary
}

\begin{document}

\maketitle
\setcounter{footnote}{0}

\begin{abstract}
    In a graph where vertices have preferences over their neighbors, a matching is called popular if it does not lose a head-to-head election against any other matching when the vertices vote between the matchings. Popular matchings can be seen as an intermediate category between stable matchings and maximum-size matchings. In this paper, we aim to maximize the utility of a matching that is popular but admits only a few blocking edges. 
    
    For general graphs already finding a popular matching with at most one blocking edge is \NP-complete. For bipartite instances, we study the problem of finding a maximum-utility popular matching with a bound on the number (or more generally, the cost) of blocking edges applying a multivariate approach. We show classical and parameterized hardness results for severely restricted instances. By contrast, we design an algorithm for instances where preferences on one side admit a master list, and show that this algorithm is optimal.
\end{abstract}

\section{Introduction} 

In the classic stable matching problem, we are given a bipartite graph, where the two sets of vertices represent two agent sets. Each agent has a strictly ordered preference list over their possible partners from the other agent set. A matching is \emph{stable} if it is not \emph{blocked} by any edge, that is, no pair of agents exists who are mutually inclined to abandon their partners for each other. The existence of stable matchings was shown in the seminal paper of Gale and Shapley~\cite{GS62}. The optimality notion was later extended to various other input settings.

In the area of matchings under preferences, the trade-off between stability and size (or utility) has been an actively investigated topic~\cite{BMM10,FKPS10,CHS+18,ABG+20,GJR+20}. The most extensively studied compromise between these optimality criteria might be the notion of popular matchings, first defined by G\"ardenfors~\cite{Gar75}. 
Matching~$M$ is more popular than another matching $M'$ if the number of vertices preferring $M$ to $M'$ is larger than the number of vertices preferring $M'$ to~$M$. A matching~$M$ is called \emph{popular} if there is no matching~$M'$ that is more popular than~$M$. 
In bipartite graphs, all stable matchings have the same size, at least~$\frac{1}{2} |M_{\text{max}}|$ where $M_{\text{max}}$ is a maximum-size matching. Stable matchings are minimum-size popular matchings~\cite{Gar75,HK13}, and maximum-size popular matchings have size at least~$\frac{2}{3} |M_{\text{max}}|$~\cite{Kav14}. In non-bipartite graphs
it is \NP-complete to decide whether a popular matching even exists~\cite{FKPZ19,GMSZ21}.
 If edges have utilities in~$\{1,2\}$, computing a popular matching of maximum utility is \NP-hard even in bipartite instances~\cite{FKPZ19}. The difference between a maximum-utility stable and a maximum-utility popular matching can be arbitrarily large in terms of the objective.


In this paper, we concentrate on the stability versus maximum utility question within the set of popular matchings. We investigate \emph{popular matchings with bounded instability}, that is, matchings that are popular but admit only a few blocking edges. Can we find popular matchings with, say, at most one blocking edge efficiently? If the given instance admits a stable matching (which is always the case for bipartite graphs), the answer is obviously yes, since all stable matchings are popular. But can we efficiently find a maximum-utility matching among all popular matchings with at most one blocking edge? 
How far can we venture into the realm of instability, if we want to keep our ability to find maximum-utility matchings efficiently?

In order to examine such questions in detail, we use a general model where with each edge we associate a \emph{utility} and a \emph{cost} as well, and we ask for a popular matching whose total utility is above a given threshold, while its blocking edges have total cost not exceeding a given budget. This setting can be interpreted as follows. Vertices are agents whose partnership brings as much profit to a central authority as the utility of the edge connecting them. The cost of an edge is the regret of the agents if the edge blocks the matching. The central authority disposes of limited resources to compensate agents who could be better off by switching to blocking edges. Up to this limit, blocking edges can be paid for and thus tolerated. The goal is to find a matching whose utility reaches our target, while ensuring that agents can be compensated from the cost budget. 

Since the question is computationally intractable in such a general form as we will see in Section~\ref{sec:prelim}, we 
apply the framework of parameterized complexity~\cite{df99}, and take a multivariate approach~\cite{Nie10} in order to understand how exactly the several parameters appearing in such an instance contribute to its intractability, and to identify cases that can be solved efficiently. Apart from the several natural parameters we can associate with the problem (such as our objective value or our budget), we also investigate various restrictions on the preference profiles and how they influence the computational complexity of the question.

\subsection{Related work}
We briefly review some classical and parameterized complexity results about stable and popular matchings. Then we elaborate on the versions of these problems with edge utilities, and finally, we discuss the relaxation of stability.
 
\paragraph{Complexity results for stable and popular matchings.}
From the seminal paper of Gale and Shapley~\cite{GS62} we know that in bipartite instances a stable matching always exists and can be found in linear time. 
If the graph is not bipartite, then the existence of a stable solution is not guaranteed.
However, Irving's linear-time algorithm finds a stable matching or reports that none exists~\cite{Irv85}. 

For bipartite instances, it was already noticed by Gärdenfors~\cite{Gar75} that all stable matchings are popular, which implies that in bipartite instances popular matchings always exist. In fact stable matchings are smallest-size popular matchings,
as shown by Biró et al.~\cite{BIM10}; maximum-size popular matchings can be found in polynomial time as well~\cite{HK13,Kav14}. Only recently Faenza et al.~\cite{FKPZ19} and Gupta et al.~\cite{GMSZ21} resolved the long-standing open question on the complexity of deciding whether a popular matching exists in a 
non-bipartite instance and showed that the problem is $\NP$-complete. 

Stable matchings have been studied extensively from a parameterized viewpoint. For an overview, please consult the survey by Chen~\cite{Che19}.
%
Only a few results consider popularity within the parameterized framework, and most of them in the context of the house allocation problem~\cite{CFP21,KKM+22}.

There is a large set of results on stable and popular matchings in instances where preferences admit a master list (see Section~\ref{sec:prelim} for a definition)~\cite{KNN14,Kam19,BHK+20,MR20}.
Master lists naturally occur in a number of applications, such as P2P networks~\cite{GLMMRV07,LMVGRM07}, job markets~\cite{IMS08}, and student housing assignments~\cite{PPR08}. 

\paragraph{Stable matchings with edge utilities.}
For bipartite instances, Irving et al.~\cite{ILG87} proposed an algorithm for finding a maximum-utility stable matching in~$O(n^4 \log n)$ time for an $n$-vertex graph; see also~\cite{GI89}. 
For non-negative integer edge utilities that satisfy a certain monotonicity requirement, 
the fastest known algorithm is due to
Feder~\cite{Fed92,Fed94}, running in 
$O(n^{2} \log({\frac{K}{n^2}}+2))\! \cdot \! \min{\{n, \sqrt{K}\}}$ time where $K$ is the utility of an optimal solution. A maximum-utility stable matching can also be computed using a 
simple and elegant formulation of the stable matching polytope~\cite{Rot92}. For the non-bipartite case, finding a maximum-utility stable matching is $\NP$-hard, but 2-approximable under certain monotonicity constraints using LP methods~\cite{TS97,TS98}.
\paragraph{Popular matchings with edge utilities.}
For bipartite instances, Faenza et al.~\cite{FK21} showed that it is \NP-complete to decide if there exists a popular matching that contains two given edges, which is a very restricted case of popular matchings with edge utilities. The same authors provided a 2-approximation algorithm for non-negative edge utilities. \NP-hardness was established for non-bipartite instances with edge utilities a couple of years earlier already~\cite{HK21}.

\paragraph{Almost stable matchings.} Non-bipartite stable matching instances need not admit a stable solution. The number of blocking edges is a characteristic property of every matching. The set of edges blocking $M$ is denoted by~$\bp(M)$. A natural goal is to find a matching minimizing~$|\bp(M)|$---such a matching is called \emph{almost stable}. This approach has a broad literature: almost stable matchings have been investigated in bipartite~\cite{KMV94,HIM09,BMM10,GJR+20} and non-bipartite instances~\cite{ABM06,BMM12,CHS+18,CIM19}.
Closest to our work is the paper by Gupta et al.~\cite{GJR+20} studying the trade-off between size and stability
from a parameterized complexity viewpoint.

\subsection{Our results and structure of the paper}

Using a multivariate approach, we gain insight into the computational complexity of finding a maximum-utility popular matching respecting a bound on the cost of its blocking edges.
We draw 
a detailed map of the problem's complexity in terms of parameters such as the cost budget~$k$, the desired utility value~$t$, the form of the cost and utility functions (e.g., being binary or uniform), and the structural properties of the preference profile; see Table~\ref{tab:summary} for a detailed summary. The proof of all results marked by an asterisk~($\star$) can be found in Appendix~\ref{sec:appendix}. Whenever possible, we provide a high-level proof sketch in the body of the paper instead.

\paragraph{Sections~\ref{sec:prelim} and \ref{sec:hardness}.}
We first define our model, optimality notions, and problems.
In Observation~\ref{thm:generalG-k=1} we prove that for general graphs, already finding \emph{any} popular matching with at most one blocking edge is \NP-complete. We thus restrict ourselves to the bipartite case and show in Theorems~\ref{thm:k=1-bounded} and~\ref{thm:k=1-maxsize-bounded} that finding a maximum-utility popular matching with at most one blocking edge is \NP-complete even for highly restricted inputs. 

\paragraph{Section~\ref{sec:masterlist}.}
To contrast these strong intractability results, we next focus on the ``tractability island'' of bipartite instances that admit a master list on one side. We propose a simple algorithm that finds a maximum-utility popular matching whose blocking edges have total cost at most~$k$, given a positive integer cost function on the edges. Our algorithm runs in time $O(|E|^k)$ where $E$ is the edge set of the input graph (Theorem~\ref{thm:mupmu-master}). This running time is tight in the sense that the problem is $\mathsf{W}[1]$-hard with parameter $k$ (Theorem~\ref{thm:masterlist-W1-hard}). We show our algorithm's optimality also in the sense that the few assumptions we have on the input (besides admitting a master list), namely that preferences are strict and  edges have positive integer costs, are necessary: allowing for ties or for zero-cost edges undermines the tractability of the problem (Theorems~\ref{thm:sociallystable} and~\ref{thm:masterlists-ties}).

\paragraph{Section~\ref{sec:pareto}.}
We close our investigations with relaxing the requirement of popularity, and focusing on the less restrictive requirement of
Pareto-optimality instead. Hence, we ask the following question: can we efficiently find a maximum-utility Pareto-optimal matching with only a few blocking edges or, more generally, one whose blocking edges have  total cost not exceeding a given budget? In Theorem~\ref{thm:t1-xpalgo} we propose an algorithm for this problem that runs in $2^{O(k \log k)}|E|^{k+2}$ time, for a positive integer  cost function on the edge set~$E$ and a budget~$k$. We also prove that this algorithm is essentially optimal (Corollary~\ref{cor:mupomu-t1}).


\section{Preliminaries}
\label{sec:prelim}

We first introduce our model and the most important concepts in Section~\ref{sec:notation}, and 
then state our problem definitions and start our investigations in Section~\ref{sec:problemdef}. 

\subsection{Model and notation}
\label{sec:notation}


\paragraph{Graphs.} 
For a graph $G=(V,E)$, we let $V(G)$ and $E(G)$ denote its vertex and edge set, respectively. If $F \subseteq E$, then $V(F)$ is the set of all endpoints in~$F$. All our graphs are simple (without loops or parallel edges). For a vertex $v \in V$, $N_G(v)$ denotes the set of its \emph{neighbors} and $\delta_G(v) = |N_G(v)|$ its \emph{degree} in~$G$. The maximum degree of $G$ is $\Delta_G = \max_{v \in V(G)} {\delta_G(v)}$.
Two edges are \emph{adjacent}, if they share an endpoint.
A \emph{matching} in $G$ is a set of edges such that no two of them are adjacent. For a matching~$M$ and an edge~$(a,b) \in M$, we let $M(a)=b$, and conversely, $M(b)=a$. 
For a set~$X$ of edges or vertices in~$G$, let~$G-X$ be the subgraph of~$G$ obtained by deleting~$X$ from~$G$; for a singleton $X=\{x\}$ we may simply write $G-x$. 
We also let $G[X]=G-(V(G) \setminus X)$ for some $X \subseteq V(G)$.

\paragraph{Preference systems.}
A \emph{preference system} is a pair $(G,\preceq)$ where $G=(V,E)$ is a graph and $\preceq$ is a collection of preference orders $\preceq_v$ for each $v \in V$, where $\preceq_v$ can be a strict or a weak linear order over $N(v)$. 
We let $\prec_v$ be the strict part of~$\preceq_v$, i.e., $a \prec_v b$ means $b \not\preceq_v a$. We say that $v$ \emph{prefers}~$a$ to~$b$ if~$b \prec_v a$, and $v$ \emph{weakly prefers}~$a$ to~$b$ if $b \preceq_v a$.
Mostly we will deal with \emph{strict preference systems} where $\preceq_v$ 
is a strict linear order for each $v \in V$; in this case we write~$(G,\prec)$. 
We say that $(G,\preceq)$ is \emph{complete}, if $G$ is a complete graph or, if we assume a bipartite setting, where $G$ is a complete bipartite graph.
We will say that $(G,\preceq)$ is \emph{compatible} with a complete preference system $(G',\preceq')$ if $G$ is a subgraph of $G'$, and for any $v \in V(G)$, the restriction of $\preceq'_v$ to $N_G(v)$ is exactly~$\preceq_v$.
For a set~$X$ of edges or vertices, we define the deletion of~$X$ from~$(G,\preceq)$ as $(G-X,\preceq^{G-X})$ where $\preceq^{G-X}$ is the restriction of~$\preceq$ to~$G-X$, containing for each~$v \in V(G-X)$ a preference order~$\preceq^{G-X}_v$ over~$N_{G-X}(v)$.

\paragraph{Stability, popularity, Pareto-optimality.}
Given a preference system \mbox{$(G,\preceq)$} and a matching~$M$ in~$G$, some~$(a,b) \in E$ is a \emph{blocking edge for $M$} 
if $a$ is unmatched or prefers~$b$ to~$M(a)$, and $b$ is unmatched or prefers~$a$ to~$M(b)$; we denote by~$\bp_G(M)$ the set of blocking edges for~$M$ in~$G$.
If~$G$ is clear from the context, we may omit the subscript (also from notations $\delta_G(\cdot)$ or $N_G(\cdot)$).
We say that $M$ is \emph{stable} in $G$ if $\bp_G(M)=\emptyset$.\footnote{When $(G,\preceq)$ is not strict,  stability as we define it is often called \emph{weak stability}. 
See the book~\cite{Man13} for other stability notions for weakly ordered preferences.}

For two matchings~$M$ and~$M'$ in~$G$, some vertex $v$ \emph{prefers}~$M$ over~$M'$, if either $v$ is matched in~$M$ but unmatched in~$M'$, or $M'(v) \prec_v M(v)$.
We say that $M$ is \emph{more popular} than~$M'$, if more vertices prefer~$M$ to~$M'$ than vice versa. A matching~$M$ is \emph{popular}, if no matching is more popular than~$M$. 

A \emph{Pareto-improvement} of a matching $M$ is a matching~$M'$ such that no vertex prefers $M$ to $M'$, and at least one prefers $M'$ to $M$.
A matching $M$ is \emph{Pareto-optimal},
if there is no Pareto-improvement for it.
Although stability, popularity and Pareto-optimality are defined in the context of a preference system~$(G,\preceq)$, when~$\preceq$ is clear from the context, we may simply say that a matching is stable, popular, or Pareto-optimal in~$G$. Notice that for strict preference systems, stable matchings are popular, and popular matchings are Pareto-optimal.

\paragraph{Structured preferences.}
In a bipartite preference system \mbox{$(G=(A,B;E),\preceq)$}, 
a \emph{master list} over vertices of~$A$ is defined as an ordering~$\mathcal{L}_A$ of all vertices in~$A$ such that for any~$b \in B$, restricting~$\mathcal{L}_A$ to $N_G(b)$ yields exactly the ordering~$\preceq_b$.
We say that $(G,\preceq)$ \emph{admits a master list on one side}, if there exists a master list over either~$A$ or~$B$; if both holds, then $(G,\preceq)$ \emph{admits a master list on both sides}.

Single-peaked and single-crossing preferences have been defined for certain stable matching problems as well~\cite{BM86,BCF+20}, but they originate from problems in the context of elections, where preferences are complete linear orders.
We follow standard definitions from the social choice literature 
that adapt these notions to incomplete preferences~\cite{EFLO15,BCFN20,FL20}.
For simplicity, let us assume that $(G,\prec)$ is a strict preference system. 
Then $(G,\prec)$ has \emph{single-peaked} preferences, if there exists a strict linear ordering~$\triangleright$ of all vertices in $V(G)$ called an \emph{axis} such that for any vertex $v \in V(G)$ and
for every $a,b,c \in N_G(v)$ with $a \triangleright b \triangleright c$, the relation $b \prec_v a$ implies $c \prec_v b$.
In this case we also say that $(G,\prec)$ is \emph{single-peaked with respect to the axis~$\triangleright$}.
If $G=(A,B;E)$ is bipartite, then it suffices to provide a suitable axis for $A$ and for~$B$ separately. 

To define single-crossing preferences, let us first assume that $(G,\prec)$ is a strict and complete preference system. For any~$a,b \in V(G)$ let $V^{a \prec b}=\{v \in V(G): a \prec_v b\}$ denote the set of vertices preferring~$b$ to~$a$. 
We say that $(G,\prec)$ is \emph{single-crossing with respect to a strict linear ordering~$\triangleright$} of~$V(G)$, if for any $a,b \in V(G)$ either all vertices in~$V^{a \prec b}$ precede all vertices in~$V^{b \prec a}$ according to~$\triangleright$, or just the opposite, all vertices in~$V^{b \prec a}$ precede all vertices in~$V^{a \prec b}$ according to~$\triangleright$. 
We say that $(G,\prec)$ is \emph{single-crossing}, if it is single-crossing with respect to some strict linear ordering of~$V(G)$.
An incomplete strict preference system is \emph{single-crossing}, if it is compatible with a complete single-crossing preference system. 
Note that if~$G$ is bipartite, then it suffices to provide a complete bipartite preference system compatible with $(G,\prec)$ and separate linear orders for~$A$ and for~$B$. 

\paragraph{Classical and parameterized complexity.} 
We assume that the reader is familiar with basic notions and techniques of classical and parameterized complexity theory; 
for an introduction or for definitions we refer to the books~\cite{GJ79,CyganEtAl2015,Nie-book}.

\subsection{Problem definitions and initial results}
\label{sec:problemdef}
Let us now formally define the problem whose computational complexity is the main focus of our  paper.

\begin{center}
\fbox{ 
\parbox{14.5cm}{
\begin{tabular}{l}\MUPMU{}:  \end{tabular} \\
\begin{tabular}{p{1.8cm}p{12.0cm}}
Input: & A strict preference system $(G,\prec)$, a
utility function~$\omega:E(G) \rightarrow \mathbb{N}$, 
a cost function~$c:E(G) \rightarrow \mathbb{N}$, an objective value~$t \in \mathbb{N}$, and a budget~$k \in \mathbb{N}$. \\
Question: & Is there a popular matching in $G$ whose utility is at least~$t$ and whose blocking edges have  total cost at most~$k$?
\end{tabular}
}}
\end{center}

For a set $F \subseteq E$ of edges in~$G$, let 
$\omega(F)=\sum_{e \in F} \omega(e)$ and $c(F)=\sum_{e \in F} c(e)$ be its utility and cost, respectively.
A matching~$M$ in~$G$ is \emph{feasible}, if both $\omega(M) \geq t$ and $c(M) \leq k$ hold.

A very natural special case of 
the above problem is when we simply  limit the number of blocking edges: this amounts to setting all edge costs to~$1$. Hence, we  are looking for a popular matching (in the hope of finding matchings with greater utility or size when compared to stable matchings) while also setting an upper bound on the instability of the matching. 
\begin{center}
\fbox{ 
\parbox{14.5cm}{
\begin{tabular}{l}\MAXUP{}:  \end{tabular} \\
\begin{tabular}{p{1.8cm}p{12.0cm}}
Input: & A strict preference system $(G,\prec)$ and an integer~$k$. \\
Question: & Is there a popular matching $M$ in $G$ with $|\bp_G(M)| \leq k$?
\end{tabular}
}}
\end{center}
Our first result shows that this problem is \NP-complete even for $k=1$. 
\begin{observation}
\label{thm:generalG-k=1}
The \MAXUP{} problem is \NP-complete for $k=1$. 
\end{observation}

\begin{proof}
See the hardness proof in \cite[Section 5.3]{CFK+22}. There is no stable matching in the constructed instance, and if there exists a popular matching $M$ then there is a unique blocking edge $(r,r')$ to~$M$.
\qed
\end{proof}

\begin{corollary}
\label{cor:general}
The \MUPMU{} problem is \NP-complete for any fixed utility function~$\omega:E \rightarrow \mathbb{N}$, 
even if the objective value is~$t=0$, the cost function is $c \equiv 1$, and our budget is $k=1$. 
\end{corollary}

Motivated by this strong intractability, in the remainder we focus on the case where the graph is bipartite. 
In this case, interestingly, the \MUPMU{} problem has a strong connection to the problem of finding a popular matching with a fixed set of blocking edges.

\begin{center}
\fbox{ 
\parbox{14.5cm}{
\begin{tabular}{l}\FUP{}:  \end{tabular} \\
\begin{tabular}{p{1.8cm}p{12.0cm}}
Input: & A strict preference system $(G,\prec)$ 
and a subset $S \subseteq E(G)$ of edges. \\
Question: & Is there a popular matching $M$ in $G$ such that $\bp_G(M)=S$?
\end{tabular}
}}
\end{center}

On the one hand, a natural idea for finding maximum-utility matchings with bounded instability is to ``guess'' the set~$S$ of blocking edges, and search for a maximum-utility matching only among matchings~$M$ for which $\bp(M)=S$, using some structural insight.
On the other hand, given an instance of \FUP{} where our aim is to ensure that each edge in~$S$ is a blocking edge, 
a possible approach is to define a utility function 
that enforces certain edges in the neighborhood of the vertices covered by $S$ to be included in any feasible matching $M$, thus yielding $\bp(M)=S$.  
This intuitive two-way connection between the two problems can be observed in the details of our reductions in Section~\ref{sec:hardness}, and will also form the basis of our positive results in Section~\ref{sec:masterlist}.

\section{Hardness results}
\label{sec:hardness}

We start with the main result of this section: Theorem~\ref{thm:k=1-bounded} shows that both \MUPMU{} and \FUP{} are computationally intractable even if the input graph~$G$ is bipartite, $\Delta_G = 3$, and preferences are single-peaked and single-crossing. 
The hardness of \MUPMU{} holds even in the following very restricted setting. 

\begin{reptheorem}{repthm_k=1bounded}[$\star$
]
\label{thm:k=1-bounded}
The \MUPMU{} problem is \NP-complete even if  
\begin{itemize}
    \item \vspace{-3pt} 
    the input graph $G$ is bipartite with $\Delta_G=3$, 
    \item preferences are single-peaked and single-crossing,
    \item the cost function is $c \equiv 1$,
    \item the budget is $k=1$, 
\item the utility function is  $\omega(e)=\left\{
\begin{array}{ll}
1 & \textrm{ if $e = f^\star$}, \\
0 & \textrm{ otherwise,}
\end{array}
\right. \phantom{i}$ for some $f^\star \in E(G)$, and
    \item the objective value is $t=1$.
\end{itemize}
\vspace{-3pt}
The \FUP{} problem is \NP-complete even if 
\begin{itemize}
\item the input graph $G$ is bipartite,
\item $\Delta_G=3$, preferences are single-peaked and single-crossing, and 
\item $|S|=1$ for the set $S$ of blocking edges.
\end{itemize}
\end{reptheorem}

Since we can test feasibility and popularity in polynomial time~\cite{HK13}, both problems are in \NP.
To prove $\mathsf{NP}$-hardness, we provide a reduction from \myproblem{Exact-3-SAT},  
the variant of 3-SAT where each clause contains exactly three literals. The correctness of our reduction heavily relies on Observation~\ref{obs:characterization} below which characterizes popular matchings with only a single blocking edge,
and can be thought of as a reformulation of the characterization of popular matchings~\cite{HK13}.

To state Observation~\ref{obs:characterization}, we need some further notation.
Given a matching~$M$ in a graph~$G$, the subgraph~$G_M$ of~$G$ is obtained by deleting those edges~$(a,b)$ outside~$M$ from~$G$ where both~$a$ and~$b$ prefer their partner in~$M$ to each other.

\begin{observation}
\label{obs:characterization}
Given a strict preference system $(G,\prec)$ where $G=(A,B;E)$ is bipartite, a matching $M$ in $G$ is a popular matching with exactly one blocking edge $e=(u,v) \in E$ if and only if the following conditions hold: 
\begin{itemize}
\addtolength{\itemindent}{6pt}
\item[\textup{(c1)}] \vspace{-3pt}
$e$ blocks $M$ in $G$; 
\item[\textup{(c2)}] $M$ is a stable matching in $G-e$;
\item[\textup{(c3)}] there exists no $M$-alternating path in $(G-e)_M$
\begin{itemize}
\addtolength{\itemindent}{22pt}
\item[\textup{(c3/i)\phantom{i}}] 
from a vertex unmatched by $M$ to $u$ or to $v$, having even length, or 
\item[\textup{(c3/ii)}] 
from $u$ to $v$, starting and ending with an edge of $M$.
\end{itemize}
\end{itemize}
\end{observation}

\noindent
{\it Proof sketch for Theorem~\ref{thm:k=1-bounded}.}
Given an instance $\varphi$ of \myproblem{Exact-3-SAT}, we construct a bipartite graph $G$ with only a single edge $f^\star$ of non-zero utility, not contained in any stable matching, so that achieving the target utility  requires the inclusion of $f^\star$ in the matching. However, a matching containing $f^\star$ will turn another edge~$e^\star$, adjacent to $f^\star$, into a blocking edge, thus introducing instability. 
In view of Observation~\ref{obs:characterization}, the computational hardness of finding a matching that contains $f^\star$ and is blocked only by $e^\star$ lies in ensuring conditions (c3/i) and (c3/ii) while avoiding the emergence of additional blocking edges.

In our construction, each clause $c$ and each variable $x$ in $\varphi$ is represented by a vertex~$a^c$ and~$b^x$, respectively, left unmatched by any stable matching of $G$.
Moreover, each literal $\ell$ in some clause $c$ is represented by a cycle $C^{x,\ell}$ in $G$, and similarly, the possible values \texttt{true} and \texttt{false} of some variable $x$ are represented by cycles $C^{x,\TT}$ and $C^{x,\FF}$, respectively.
Next, for each clause~$c$ we create a path that leads from~$a^c$ to the blocking edge~$e^\star$ and goes through certain edges of the three cycles corresponding to the three literals of~$c$.
Similarly, for each variable~$x$ we create a path that leads from~$b^x$ to~$e^\star$ and goes through certain edges of the cycles~$C^{x,\TT}$ and~$C^{x,\FF}$.

The preferences of the vertices along these paths and cycles are defined in a way that ensuring condition (c3/i) requires the desired matching $M$ to have the property that $M \triangle M_0$ contains certain cycles, for some fixed stable matching $M_0$ (defined as part of the construction).
Namely, condition (c3/i) for the unmatched vertex~$a^c$ corresponding to some clause $c$ guarantees that $M \triangle M_0$ contains at least one cycle $C^{c,\ell}$, corresponding to a literal of $c$.
Similarly, condition (c3/i) for $b^x$ for some variable $x$ guarantees that either $C^{x,\TT}$ or $C^{x,\FF}$ is contained in $M \triangle M_0$. 
The crux of the construction is the addition of so-called consistency edges running between each cycle $C^{c,\ell}$ corresponding to a literal~$\ell$ of some clause~$c$ and the cycle $C^{x,\TT}$ in case $\ell$ is the positive literal of variable~$x$, or the cycle $C^{x,\FF}$ in case $\ell$ is the literal $\overline{x}$. 
The condition that no such consistency edge can block the desired matching~$M$  ensures that the truth assignment encoded by~$M$ indeed satisfies each clause, and vice versa, any truth assignment that satisfies $\varphi$ implies a matching~$M$ that is popular, contains~$f^\star$, and is blocked only by the edge~$e^\star$.

The technical details of incorporating the above ideas into a carefully designed construction, as well as the additional ideas and arguments that ensure the required structural properties of the preferences, can be found in Appendix~\ref{sec:app-hardness}.
\hfill \qed
\medskip

We also investigate whether \MUPMU{} becomes easier if we do not allow edges with zero utility: 
Theorem~\ref{thm:k=1-maxsize-bounded} below shows that if the utility function is $\omega \equiv 1$ and we aim for a popular matching with only one more edge than a stable matching, 
the problem is \NP-complete even if all other restrictions of Theorem~\ref{thm:k=1-bounded} remain in place.

Theorem~\ref{thm:k=1-maxsize-bounded} can be obtained by combining the reduction proving  Theorem~\ref{thm:k=1-bounded}
with ideas used in a similar reduction in~\cite[Section 5.1]{CFK+22};
see Remark~\ref{remark:altproof} in Appendix~\ref{sec:app-hardness} for more details.

\begin{reptheorem}{repthm_k=1maxsize}[$\star$]
\label{thm:k=1-maxsize-bounded}
The \MUPMU{} problem is \NP-complete even if  
\begin{itemize}
    \item \vspace{-3pt} 
    the input graph $G=(V,E)$ is bipartite and
     $\Delta_G = 3$, 
    \item preferences are single-peaked and single-crossing,
    \item the cost function is $c \equiv 1$,
    \item the budget is $k=1$, 
    \item the utility function is $\omega \equiv 1$, and
    \item the objective value is $t=|M_s|+1=\frac{|V|}{2}$ where $M_s$ is a stable matching in~$G$.
\end{itemize}
\end{reptheorem}

\begin{corollary}
\label{cor:maxsize}
Given a preference system~$(G,\prec)$, finding a popular matching in~$G$ with at most one blocking edge  that is larger than a stable matching is \NP-complete, even if preferences in $(G,\prec)$ are single-peaked and single-crossing.
\end{corollary}

\section{Algorithms for preferences admitting a master list}
\label{sec:masterlist}

In this section we focus on the case when preferences admit a master list on one side. It is known 
that strict preference systems with this property admit a  unique stable matching~\cite{IMS08}; this fact is the backbone of our simple approach.

\begin{theorem}
\label{thm:fup-master}
An instance $(G,\prec,S)$ of \FUP{} where the input graph $G=(A,B;E)$ is bipartite and $(G,\prec)$ admits a master list on one side 
can be solved in $O(|E|)$ time.
\end{theorem}

\begin{proof}
Let $M$ be a popular matching in $G$ with $\bp_G(M)=S$.
Then $M$ is stable in the graph $G' = G-S$.
Since $(G,\prec)$ admits a master list on one side,
$G'$ admits a unique stable matching $M'$ which can be found in $O(|E|)$ time~\cite{IMS08}.
If $M'$ is not popular in $G$ or $\bp_G(M) \neq S$, we output `No', otherwise we output~$M'$.
Finding~$M'$ in~$G'$, computing $\bp_G(M')$ and testing whether $M'$ is popular in~$G$ can all be done in $O(|E|)$ time (for testing popularity, see~\cite{HK13}).
\qed 
\end{proof}

\begin{theorem}
\label{thm:mupmu-master}
An instance $(G,\prec,\omega,c,t,k)$ of \MUPMU{} where the input graph $G=(A,B;E)$ is bipartite,  $(G,\prec)$ admits a master list on one side, and $c(e) \geq 1$ for all edges $e \in E$ can be solved in~$O(|E|^{k+1})$ time.
\end{theorem}

\begin{proof}
Since all edges have cost at least~$1$, the desired matching may admit at most $k$ blocking edges. 
As observed in the proof of Theorem~\ref{thm:fup-master}, for a given subset~$S \subseteq E$ of edges, there exists at most one  matching~$M$ with $\bp_G(M)=S$.
By trying all edge sets $S \subseteq E$ of size at most~$k$, we can check all possible solutions in time $O(|E|^{k+1})$, since we can check feasibility in linear time.
\qed
\end{proof}

We contrast Theorem~\ref{thm:mupmu-master} by showing that the running time $O(|E|^{k+1})$ is optimal in the sense that we cannot expect an algorithm that runs in $f(k) \cdot |E|^{O(1)}$ time for some function~$f$, i.e., an FPT algorithm with parameter~$k$, as the problem is $\mathsf{W}[1]$-hard with parameter~$k$, even for~$t=1$.

\begin{reptheorem}{repthm_ML_W1hard}[$\star$]
\label{thm:masterlist-W1-hard}
The \MUPMU{} problem is \NP-complete and $\mathsf{W}[1]$-hard with parameter $k$, even if 
\begin{itemize}
    \item \vspace{-3pt} the input graph~$G$ is bipartite,
    \item preferences admit a master list on both sides,  
    \item the cost function is $c \equiv 1$, and
\item either \emph{(a)} the utility function  is  $\omega(e)=\left\{
\begin{array}{ll}
1 & \textrm{ if $e = f^\star$}, \\
0 & \textrm{ otherwise,}
\end{array}
\right.$ for some $f^\star \in E(G)$ \\
$\phantom{either \emph{(b)} i}$ with $\omega(M_s)=0$, and the objective value is $t=1$, or \\
    $\phantom{either}$ \emph{(b)}
    $\omega \equiv 1$ and   $t=|M_s|+1=|V(G)|/2$, \\
    where $M_s$ is the unique stable matching in $G$.\end{itemize}
\end{reptheorem}

To prove Theorem~\ref{thm:masterlist-W1-hard}, we need a different approach than the one used to prove the results of Section~\ref{sec:hardness}. 
There is a simple reason why preferences in those constructions do not admit a master list (on either side): vertices within cycles corresponding to literals or to truth assignments of a variable are cyclic (in the sense that each vertex on such a cycle prefers the ``next'' vertex along the cycle to the ``previous'' vertex on the cycle, when traversing the cycle in one direction), and hence do not admit a master list on either side. 

The reduction proving Theorem~\ref{thm:masterlist-W1-hard} is from the $\mathsf{W}[1]$-hard \myproblem{Multicolored Clique} problem, 
and although it applies standard techniques from the literature (e.g., a similar approach is used in \cite{GJR+20}), proving its correctness requires detailed arguments.  
To give the reader some intuition about the workings of the reduction, we present a short sketch.

\medskip
\noindent
{\it Proof sketch for Theorem~\ref{thm:masterlist-W1-hard}.}
The input of  \myproblem{Multicolored Clique} is a graph $G=(V,E)$ with~$V$ partitioned into sets $V_1,\dots,V_q$, and an integer parameter~$q$, and the task is to decide whether $G$ contains a clique of size~$q$ containing exactly one vertex from each of the sets~$V_i$. 

We introduce a vertex gadget $G_i$ for each $i \in [q]$, and an edge gadget $G_{i,j}$ for each $\{i,j\} \subseteq [q]$ with $i<j$. Both vertex and edge gadgets will have the same underlying structure, 
which we present now for $G_i$: It contains an edge $(a_v,b_v)$ for each $v \in V_i$, a ``source'' vertex~$s_i$ and a ``sink'' vertex~$t_i$.\footnote{We only use the terms ``source'' and ``sink'' for illustration; the constructed graph is undirected.}
We connect $s_i$ to each vertex in $\{a_v:v \in V_i\}$, and analogously,  we connect $t_i$ to each vertex in $\{b_v:v \in V_i\}$. Edge gadgets are defined analogously, with $G_{i,j}$ containing edges of the form~$(a_e,b_e)$ for each edge $e \in E$ running between $V_i$ and~$V_j$.

The next idea is to connect all gadgets, one after the other, threading them along a path by  connecting the sink vertex of each gadget with the source vertex of the next gadget (see Figure~\ref{fig:2ML-hardness} in Appendix~\ref{sec:app-masterlist} for an illustration). 
Preferences are defined so that the unique stable matching $M_s$ contains all edges of the form $(a_v,b_v)$ and $(a_e,b_e)$, and leaves unmatched only the source vertex of the first gadget, say $s_1$, and the sink vertex of the last gadget, say $t_{q-1,q}$. Thus, a desired matching~$M$ will be such that $M \triangle M_s$ is a path $P_M$ from $s_1$ to $t_{q-1,q}$ that traverses all gadgets. The edges of $P_M$ used within the gadgets will therefore correspond to selecting $q$ vertices and $\binom{q}{2}$ edges in the input graph~$G$. 
Additional inter-gadget edges, running between vertex and edge gadgets,
will ensure that $M$ is a popular matching of size $|M_s|+1$ with $|\bp(M)|\leq q+\binom{q}{2}$ exactly if the endpoints of any edge ``selected'' by $P_M=M_S \triangle M$ are also ``selected'', which in turn can happen exactly if $G$ admits a clique of size~$q$ in $G$ as required.

The definitions for the preferences in the construction as well as the precise arguments for its correctness can be found in Appendix~\ref{sec:app-masterlist}.
\qed

\medskip
\paragraph{Free edges.}
One may wonder if the restriction in Theorem~\ref{thm:mupmu-master} that all edge costs are at least~1 is necessary, or our algorithm can be extended to accommodate edges of cost~0. Such edges are called \emph{free edges} in the literature, and a matching whose blocking edges are all free edges is called a \emph{socially stable} matching.
It is known that
deciding whether there exists a complete socially stable matching in a bipartite instance with master lists on both sides is \NP-complete~\cite[Theorem~5.3.4]{Kwa15}; however, this does not imply the following theorem where we prove that it is \NP-hard to find a complete \emph{popular} socially stable matching.
The proof of Theorem~\ref{thm:sociallystable} is a slight modification of the reduction proving Theorem~\ref{thm:masterlist-W1-hard}.

\begin{reptheorem}{repthm_sociallystable}[$\star$]
\label{thm:sociallystable}
The \MUPMU{} problem is \NP-complete  
even if 
\begin{itemize}
    \item \vspace{-4pt} the cost function $c$ is binary, 
    \item the budget is $k=0$, and 
    \item all other restrictions of Theorem~\ref{thm:masterlist-W1-hard} hold.
\end{itemize}
\end{reptheorem}


\paragraph{Ties in the master list.}
Another way of generalizing Theorem~\ref{thm:fup-master} would be to extend its result to the case where preferences are not necessarily strict but may include ties.
We define the \MUPMUT{} problem the same way as its strict variant, with the only difference that the input preference system may not be strict; we define \FUPT{} analogously.

Our next result shows that there is no hope that the algorithm of Theorem~\ref{thm:fup-master} can be extended to the case where we allow ties in the preference lists, even if we further require severe restrictions on the input. 
The proof of Theorem~\ref{thm:masterlists-ties} is based on the reduction proving Theorem~\ref{thm:k=1-bounded}, and its main idea is to circumvent the problem of cyclic preferences in the construction by using ties.

\begin{reptheorem}{repthm_MLties}[$\star$]
\label{thm:masterlists-ties}
The \MUPMUT{} problem is \NP-complete even if
\begin{itemize}
    \item \vspace{-3pt} the input graph $G=(A,B;E)$ is bipartite and $\Delta_G = 3$,
    \item preferences on both sides admit a master list, 
    \item the cost function is $c \equiv 1$, 
    \item the budget is $k=1$, 
    \item \vspace{-3pt} the utility function is  $\omega(e)=\left\{
\begin{array}{ll}
1 & \textrm{ if $e = f^\star$}, \\
0 & \textrm{ otherwise,}
\end{array}
\right. \phantom{i}$ for some $f^\star \in E(G)$, and
    \item \vspace{-1pt} the objective value is $t=1$.
\end{itemize}
 \vspace{-3pt}
The \FUPT{} problem 
is \NP-complete even if the input graph~$G$ is bipartite,  $\Delta_G = 3$,  preferences  on both sides admit a master list, and $|S|=1$ for the set $S$ of blocking edges.
\end{reptheorem}

\section{Pareto-optimal matchings with bounded instability}
\label{sec:pareto}
In this section we shift our attention to the following problem. 

\begin{center}
\fbox{ 
\parbox{14.5cm}{
\begin{tabular}{l}\MUPOMU{}:  \end{tabular} \\
\begin{tabular}{p{1.8cm}p{12.0cm}}
Input: & A strict preference system $(G,\prec)$, 
a utility function~$\omega:E(G) \rightarrow \mathbb{N}$, 
a cost function~$c:E(G) \rightarrow \mathbb{N}$, and two integers~$t$ and~$k$. \\
Question: & Is there a Pareto-optimal matching in $G$ whose utility is at least~$t$ and whose blocking edges have  total cost at most~$k$?
\end{tabular}
}}
\end{center}

A natural approach to solve this problem is  to
guess the set $S$ of blocking edges in~$G$, and find a  stable matching $M$ of maximum utility in the graph~$G-S$. 
But even though $M$ is Pareto-optimal in~$G-S$, it may not be Pareto-optimal in~$G$ (see Example~\ref{example:PO} in Appendix~\ref{sec:app-pareto}). We propose the following approach~instead.


\paragraph{Our algorithm.}

Let $M$ be a solution for our input instance~$I=(G,\prec,\omega,c,t,k)$. 
Consider the set $F_M$ of edges that run between vertices in $V(\bp_G(M))$ but do \emph{not} belong to~$M$. We will call the pair $(\bp_G(M),F_M)$ the \emph{hint} for $M$. 

Let $\mathcal{H}_I$ denote the set of all pairs $(S,F)$ where $S \subseteq E$ with~$|S| \leq c(S) \leq k$, and $F$ is a subset of the edges in $G[V(S)]-S$
such that $G[V(S)] - (S \cup F)$ is a matching. 
Our algorithm iterates over each~$H \in \mathcal{H}_I$, and searches for a matching $M_H$ whose hint is $H$.
With each hint~$H$, we associate a weight function $w_H:E\setminus S \rightarrow \mathbb{N}$ 
as follows: 

\vspace{-8pt}
\[
w_H(e)= \left\{
\begin{tabular}{ll}
0 & \textrm{ if $e \in F$,} \\
$\omega(e)+w_0 \quad$ & \textrm{ if $e \in E \setminus (S \cup F)$,}
\end{tabular}
\right.
\vspace{-4pt}
\]
where $w_0$ is large enough to ensure that any maximum-weight stable matching in~$G-S$ is disjoint from~$F$; 
setting $w_0=|M_s| \cdot \max_{e \in E}\{\omega(e)\}$ suffices for our purpose, where $M_s$ is any stable matching in~$G-S$.

\vspace{-14pt}
\begin{algorithm}[h]
\caption{ \hspace{10pt} Input: $I=(G,\prec,\omega,c,t,k)$.}
\label{alg:PO}
\begin{algorithmic}[1]
\ForAll{$H=(S,F) \in \mathcal{H}_I$} 
\State Compute a maximum-weight stable matching $M_H$ w.r.t.~$w_H$ in $G-S$. 
\If{$M_H$ is feasible and Pareto-optimal in $G$} \textbf{return} $M_H$. \EndIf
\EndFor
\State \textbf{return} ``No feasible Pareto-optimal matching exists for $I$.''  
\end{algorithmic}
\end{algorithm}




\begin{reptheorem}{repthm_POalgo}
\label{thm:t1-xpalgo}
Algorithm~\ref{alg:PO} solves an instance $(G,\prec,\omega,c,t,k)$ of \MUPOMU{}  where $G=(V,E)$ is bipartite 
and $c(e) \geq 1$ for all edges $e \in E$ in $2^{O(k \log k)} |E|^{k+2}$ time. 
\end{reptheorem}

\begin{proof}
Note that whenever Algorithm~\ref{alg:PO} outputs a matching, it is clearly Pareto-optimal and feasible in~$G$ (since it checks both).
So it remains to prove that whenever there exists a feasible Pareto-optimal matching $M$ in $G$, Algorithm~\ref{alg:PO} finds one. Let $H=(S_M,F_M)$ be the hint of $M$; clearly, $H \in \mathcal{H}_I$.
We claim that $M_H$ is also feasible and Pareto-optimal (though $M \neq M_H$ is possible). 

First, since $M_H$ is stable in $G-S_M$, we get $\bp_G(M_H) \subseteq \bp_G(M)=S_M$ and so $c(\bp_G(M_H)) \leq c(\bp_G(M)) \leq k$.

Second, to see $\omega(M_H) \geq \omega(M)\geq t$, note that $w_H(M_H) \geq w_H(M)$ by the choice of $M_H$. Since $M$ has weight $|M_s|\cdot w_0 + \omega(M)$, 
by our choice of $w_0$ it follows that $M_H$ must also contain $|M_s|$ edges whose weight (according to $w_H$) is at least $w_0$, i.e., $|M_s|$ edges in $E \setminus (S_M \cup F_M)$. 
Since $|M_s|=|M|=|M_H|$, this implies $M_H \cap F_M=\emptyset$. Thus, we get 
\[w_H(M_H)=|M_s| \cdot w_0+\omega (M_H)\geq w_H(M)=|M_s|\cdot  w_0+\omega(M),\] 
implying $\omega(M_H) \geq \omega(M)$. 
Hence, $M_H$ is feasible.

Third, we claim that $M_H$ is Pareto-optimal in $G$. 
Let us assume otherwise for the sake of contradiction. Then there exists a matching~$M'$ in~$G$ that is a Pareto-improvement over~$M_H$. 
By the definition of Pareto-improvements, any edge in~$M' \setminus M_H$ blocks $M_H$, 
so $M' \setminus M_H \subseteq S_M$. 
Let $D=M' \triangle M_H$.
Note that $D$ is the disjoint union of cycles and paths whose endpoints are both unmatched in~$M_H$ but matched in~$M'$. 
This implies that each edge in $D \cap M_H$ connects two vertices incident to some edge in $D \cap M' \subseteq M' \setminus M_H \subseteq S_M$, i.e., 
edges of~$D \cap M_H$ are in~$G[V(S_M)]$.
Recall that every edge of~$G[V(S_M)]$ that is not in~\mbox{$S_M \cup F_M$} is in~$M$ (by the definition of~$F_M$).
Therefore, as $M_H \cap F_M= \emptyset$ and $M_H$ is a matching in~$G-S_M$, we obtain $D \cap M_H \subseteq M$.  

We claim that $M \triangle D$ is a matching that is a Pareto-improvement over~$M$ in~$G$, contradicting the Pareto-optimality of $M$. 
That $M \triangle D$ is a matching follows from the facts that $D \cap M_H \subseteq M$, that $D \setminus M_H \subseteq S_M$ is disjoint from~$M$, and that $M_H$ and $M$ match the same set of vertices in $G$ (as they are both stable in $G-S_M$). 
To see that $M \triangle D$ is a Pareto-improvement over $M$, it suffices to observe that all edges in~$D \setminus M \subseteq \S_M=\bp_G(M)$ block $M$.  

Hence, we can conclude that Algorithm~\ref{alg:PO} will find at least one matching, namely $M_H$, that is feasible and Pareto-optimal in~$G$.

Regarding the running time, first observe that $|\mathcal{H}_I| \leq |E|^k \cdot k!$ because there are at most~$|E|^k$ ways to select some $S\subseteq E$ of size at most~$k$, and given~$S$, there are at most~$k!$ possibilities to select~$F$ so that $(S,F) \in \mathcal{H}_I$. The latter follows from the observation that selecting $F$ from the edges of $G[V(S)]-S$ is equivalent to selecting the remaining edges of $G[V(S)]-S$, required to form a matching; note that there are at most~$k!$ matchings in $G[V(S)]-S$.

Given some hint~$H=(S,F) \in \mathcal{H}_I$,
we use the algorithm of Irving, Leather, and Gusfield~\cite{ILG87} to find a maximum-weight stable matching in~$G-S$ with respect to~$w_H$. Their method constructs a flow network based on the so-called rotation digraph of~$G-S$; this network has $N=O(|E|)$ vertices and \mbox{$M=O(|E|)$} arcs, as can be noted by inspecting the proof of Lemma 3.3.2 in~\cite{GI89}. A detailed analysis of Irving et al.'s algorithm~\cite{CDHZ12} shows that their method runs in $O(NM)$ time plus the time necessary for computing a maximum flow in the constructed network. The latter can be done in $O(NM)$ time as well, applying Orlin's algorithm~\cite{Orlin2013} (not yet available when~\cite{ILG87} and~\cite{GI89} were published). Thus, computing $M_H$ 
takes $O(NM)=O(|E|^2)$ time, implying that the total running time of Algorithm~\ref{alg:PO} is~$k! O(|E|^{k+2})=2^{O(k \log k)}|E|^{k+2}$.
\qed
\end{proof}

Theorem~\ref{thm:t1-instability} generalizes a result in~\cite{GJR+20} which shows that finding a matching in a bipartite graph that is larger by~$t'$ than a stable matching and has at most~$k$ blocking edges is $\mathsf{W}[1]$-hard with parameter~$t'+k$; we show that the parameterized hardness holds even if $t'$ is a constant, namely $t'=1$, and preferences admit a master list. 
As a consequence we obtain Corollary~\ref{cor:mupomu-t1}, showing that we cannot hope to solve \MUPOMU{} in FPT time with parameter~$k$. 
Thus, Algorithm~\ref{alg:PO} 
is roughly optimal.

\begin{reptheorem}{repthm_PO_W1hard}[$\star$]
\label{thm:t1-instability}
Given a bipartite preference system~$(G,\prec)$ with a stable matching $M_s$,
finding a matching in $(G,\prec)$ with at most~$k$ blocking edges that is larger than $M_s$ is $\mathsf{W}[1]$-hard with parameter~$k$, even if preferences in $(G,\prec)$ admit a master list on both sides and $|M_s|=|V(G)|/2-1$.
\end{reptheorem}

\begin{repcorollary}{repcor_PO}[$\star$]
\label{cor:mupomu-t1}
The \MUPOMU{} problem is $\mathsf{W}[1]$-hard with budget~$k$ as parameter, if 
\begin{itemize}
    \item \vspace{-3pt} the input graph $G$ is bipartite,
    \item preferences admit a master list on both sides,
    \item the cost function is $c \equiv 1$, 
    \item the utility function is $\omega \equiv 1$, and
    \item the objective value is $t=|M_s|+1=\frac{|V(G)|}{2}$  for a stable matching $M_s$ of $G$.
\end{itemize}
\end{repcorollary}

\section{Conclusion}

We studied the \MUPMU{} problem, and painted a detailed landscape of its computational complexity by showing how the utility and cost functions $\omega$ and $c$, the budget~$k$, our objective value~$t$, and various restrictions on the preferences affect the tractability of the problem.
We also made a brief detour into a relaxation where the requirement of popularity is replaced by Pareto-optimality. Table~\ref{tab:summary} summarizes our results. 

An interesting open question is whether the algorithm of Theorem~\ref{thm:mupmu-master} can be extended to preference domains that are, in some sense, close to admitting a master list. For example, can we efficiently solve instances where preferences of almost all vertices admit a master list?

\begin{table}[t]
\begin{center}
\resizebox{\textwidth}{!}{
\begin{tabular}{l@{\hspace{4pt}}c@{\hspace{4pt}}c@{\hspace{4pt}}c@{\hspace{4pt}}c@{\hspace{4pt}}c@{\hspace{4pt}}c}
	      \noalign{\hrule}
   &  \multicolumn{4}{c}{input settings}            
   & \multicolumn{2}{c}{result} 
\\
 problem     & graph   & preferences & costs & utilities &  complexity & theorem \\
	      \noalign{\hrule}
	      \noalign{\hrule}
  popular & general & general & $c \equiv 1,k=1$ & $ \omega$ fixed, $t=0$  & \NP-c & Cor.\!~\ref{cor:general}  \\
  popular & bipartite, $\Delta_G=3$ & SP, SC & $c\equiv 1,k=1$ & $\omega$ almost always 0, $t=1$   & \NP-c & Thm.\!~\ref{thm:k=1-bounded}  \\
  popular & bipartite, $\Delta_G=3$ & SP, SC & $c\equiv 1,k=1$ & $\omega \equiv 1, t=|M_s|+1$   & \NP-c & Thm.\!~\ref{thm:k=1-maxsize-bounded} \\
  popular & bipartite  & 1-ML & $c \geq 1$ &  &  $O(|E|^{k+1})$ alg. & Thm~\ref{thm:mupmu-master}  \\
  popular & bipartite  & 2-ML & $c \equiv 1$ & $ \omega$ almost always~0, $t=1$  &  $\mathsf{W}[1]$-h wrt~$k$ & Thm.\!~\ref{thm:masterlist-W1-hard}  \\
  popular & bipartite  & 2-ML & $c \equiv 1$ & $ \omega \equiv 1$, $t=|M_s|+1$  &  $\mathsf{W}[1]$-h wrt~$k$ & Thm.\!~\ref{thm:masterlist-W1-hard}  \\
  popular & bipartite  & 2-ML & $c$ binary, $k=0$ & $ \omega$ almost always~0, $t=1$  &  \NP-c & Thm.\!~\ref{thm:sociallystable}  \\
  popular & bipartite  & 2-ML & $c$ binary, $k=0$ & $\omega \equiv 1$, $t=|M_s|+1$  &  \NP-c & Thm.\!~\ref{thm:sociallystable}  \\
  popular w/ ties & bipartite, $\Delta_G=3$ & 2-ML & $c\equiv 1,k=1$ & $\omega$ almost always 0, $t=1$   & \NP-c & Thm.\!~\ref{thm:masterlists-ties}  \\%
  Pareto-opt & bipartite & general & $c \geq 1$   & &  $k! \cdot O(|E|^{k+2})$ alg. & Thm.\!~\ref{thm:t1-xpalgo}  \\
  Pareto-opt & bipartite & 2-ML & $c \equiv 1$ & $\omega \equiv 1, t=|M_s|+1$   &  $\mathsf{W}[1]$-h wrt~$k$ & 
            Cor.\!~\ref{cor:mupomu-t1} \\
      \noalign{\hrule}
    \end{tabular}
    }
\end{center}
\caption{Summary of our results. 
We simply write ``popular'' (``Pareto-opt'') to denote the problem
\myproblem{Max-Utility Popular (Pareto-Optimal) Matching with Instability Costs}, and add ``w/ ties'' for the variant with ties.
SP / SC / 1-ML / 2-ML mean preferences that are single-peaked /single-crossing / admit a master list on one side / on both sides, resp. 
We describe $\omega$ as ``almost always~0'' when 
$\omega$ is binary and takes value~1 only on a single edge.
We let $M_s$ denote a stable matching in the input.
}
\label{tab:summary}
\end{table}

\subsubsection{Acknowledgment.}
This work started at the Dagstuhl Seminar on Matching Under Preferences: Theory and Practice, thanks to Sushmita Gupta who initiated a research group on almost-stable popular matchings. We are grateful for Sushmita, Pallavi Jain, and Telikepalli Kavitha for fruitful discussions and helpful comments. 
We are especially indebted to Kavitha for her many contributions to the paper such as pointing out Observation~\ref{thm:generalG-k=1}, sharing with us her idea for a reduction in Remark~\ref{remark:altproof},  and for suggesting many of the questions we studied. 
\newpage

\bibliographystyle{abbrv}
\bibliography{mybib}

\newpage

\appendix

\section{Appendix: Missing proofs}
\label{sec:appendix}
Here we provide all the proofs omitted from the main text due to lack of space, organized according to the sections from which they are omitted.

\subsection{Missing proofs from Section~\ref{sec:hardness}}
\label{sec:app-hardness}

We start with a weaker version of Theorem~\ref{thm:k=1-bounded}, whose proof provides a reduction that is the basis of all our reductions in this section. 

\begin{reptheorem}{repthm_k=1}
\label{thm:k=1} 
The \MUPMU{} problem is \NP-complete even if 
\begin{itemize}
    \item \vspace{-3pt} 
    the input graph $G$ is bipartite, 
    \item the cost function is $c \equiv 1$,
    \item the budget is $k=1$, 
    \item the utility function is  $\omega(e)=\left\{
      \begin{array}{ll}
      1 & \textrm{ if $e = f^\star$}, \\
      0 & \textrm{ otherwise,}
      \end{array}
      \right. \phantom{i}$ for some $f^\star \in E(G)$, and
    \item the objective value is $t=1$.
\end{itemize}
The \FUP{} problem is \NP-complete even if the input graph is bipartite and $|S|=1$ for the set $S$ of blocking edges. 
\end{reptheorem}


\begin{proof}
Since we can test both feasibility and popularity in polynomial time~\cite{HK13}, both problems are in \NP.
We now give a reduction from \myproblem{Exact-3-SAT}, the variant of 3-SAT where each clause contains exactly three literals. 
Let us consider an input formula $\varphi=c_1 \wedge c_2  \wedge \dots \wedge c_m$ over variables $x_1, \dots, x_n$. 

\paragraph{Construction.}
We construct an instance of \MUPMU{} with a graph $G=(A,B;E)$ as follows; see Figure~\ref{fig:hardness} for an illustration.

For each variable $x_i$, we create a 4-cycle $C^{x_i,\TT}$ on vertices $a^{x_i,\TT}_0$, $b^{x_i,\TT}_0$, $a^{x_i,\TT}_1$, and $b^{x_i,\TT}_1$, 
and another 4-cycle $C^{x_i,\FF}$ on vertices $a^{x_i,\FF}_0$, $b^{x_i,\FF}_0$, $a^{x_i,\FF}_1$, and $b^{x_i,\FF}_1$, representing the \texttt{true} and \texttt{false} value assignments for variable $x_i$, correspondingly.
Similarly, for each clause $c_j$, we construct a 4-cycle on $C^{c_j,\ell}$ on vertices $a^{c_j,\ell}_0$, $b^{c_j,\ell}_0$, $a^{c_j,\ell}_1$ and $b^{c_j,\ell}_1$ where $\ell \in [3]$\footnote{For an integer $i \in \mathbb{N}$, we let $[i]=\{1, 2, \dots, i\}$.}, representing the $\ell$-th literal in clause $c_j$. 
Furthermore, we add vertices of $\hat{A}=\{a^{c_1}, \dots, a^{c_m}\}$ corresponding to clauses, 
as well as vertices of $\hat{B}=\{b^{x_1}, \dots, b^{x_n}\}$ corresponding to variables.
We will also have four special vertices: $u$, $v$, $u'$ and $v'$. 
To define our vertex set, we let 
\begin{eqnarray*}
 A^x&=&\{a_0^{x_i,\lambda},a_1^{x_i,\lambda} : i \in [n] \textrm{ and } \lambda \in \{\TT,\FF\} \}, \\
 A^c&=&\{a_0^{c_j,\ell},a_1^{c_j,\ell} : j \in [m] \textrm{ and } \ell \in [3]\}, \\
 B^x&=&\{b_0^{x_i,\lambda},b_1^{x_i,\lambda} : i \in [n] \textrm{ and } \lambda \in \{\TT,\FF\} \}, \\
 B^c&=&\{b_0^{c_j,\ell},b_1^{c_j,\ell} : j \in [m] \textrm{ and } \ell \in [3]\}, \\
 A&=& A^x \cup A^c \cup \hat{A} \cup \{u,u'\},  \quad \textrm{and}\\
 B&=& B^x \cup B^c \cup \hat{B} \cup \{v,v'\}.
\end{eqnarray*}

To define the edge set of our bipartite graph $G=(A,B;E)$, in addition to the edges of the above defined 4-cycles, 
for each clause $c_j$ we add edges so that $G$ contains the 
path $$P_j=(u,v',a^{c_j,1}_0,b^{c_j,1}_0,a^{c_j,2}_0,b^{c_j,2}_0,a^{c_j,3}_0,b^{c_j,3}_0,a^{c_j}).$$
Similarly, for each variable $x_i$ we add edges so that $G$ contains the 
path $$Q_i=(v,u',b^{x_i,\TT}_0,a^{x_i,\TT}_0,b^{x_i,\FF}_0,a^{x_i,\FF}_0,b^{x_i}).$$
Additionally, we define a set $F$ of \emph{consistency edges}: for each clause $c_j$ and each~$\ell \in [3]$ 
we add an edge $(b^{c_j,\ell}_1,a^{x_i,\TT}_1)$ if the $\ell$-th literal in $c_j$ is $\overline{x_i}$,
and 
we add an edge $(b^{c_j,\ell}_1,a^{x_i,\FF}_1)$ if the $\ell$-th literal in $c_j$ is $x_i$. Note that we connect vertices corresponding to \emph{negative} literals with vertices corresponding to a \texttt{true} value  assignment for the given variable, and similarly,
we connect vertices corresponding to \emph{positive} literals with vertices corresponding to a \texttt{false} value  assignment for the given variable.
Finally, we create the edge~$e^\star=(u,v)$. 

We define the preference relation $\prec$ through the following preference lists, ordered in decreasing order of preference. Here and henceforth, any set in a preference list should be interpreted as a list of its elements ordered arbitrarily. 
\begin{longtable}{ll}
$u$: $v,v'$; & \\
$v$: $u,u'$; & \\
$u'$: $v,b^{x_1,\TT}_0, \dots, b^{x_n,\TT}_0$; \\
$v'$: $u,a^{c_1,1}_0, \dots, a^{c_m,1}_0$; \\
$a^{c_j,1}_0$: $b^{c_j,1}_1$, $v'$, $b^{c_j,1}_0$ & where  $j \in [m]$; \\
$a^{c_j,\ell}_0$: $b^{c_j,\ell}_1$, $b^{c_j,\ell-1}_0$, $b^{c_j,\ell}_0$ & where  $j \in [m]$ and $\ell \in \{2,3\}$; \\
$b^{c_j,\ell}_0$: $a^{c_j,\ell}_0$, $a^{c_j,\ell}_1$, $a^{c_j,\ell+1}_0 \qquad$ & where  $j \in [m]$ and $\ell  \in \{1,2\}$;\\
$b^{c_j,3}_0$: $a^{c_j,3}_0$, $a^{c_j,3}_1$, $a^{c_j}$ & where  $j \in [m]$; \\
$a^{c_j,\ell}_1$: $b^{c_j,\ell}_0$, $b^{c_j,\ell}_1$ & where  $j \in [m]$ and $\ell \in [3]$;  \\
$b^{c_j,\ell}_1$: $a^{c_j,\ell}_1$, $a^{x_i,\lambda}_1$, $a^{c_j,\ell}_0$ & where  $j \in [m]$, $\ell \in [3]$ and 
$(b^{c_j,\ell}_1,a^{x_i,\lambda}_1) \in F$; \\
$a^{x_i,\TT}_0$: $b^{x_i,\TT}_0$, $b^{x_i,\TT}_1$, $b^{x_i,\FF}_0$ & where  $i \in [n]$; \\
$a^{x_i,\FF}_0$: $b^{x_i,\FF}_0$, $b^{x_i,\FF}_1$, $b^{x_i}$ & where  $i \in [n]$; \\
$b^{x_i,\TT}_0$: $a^{x_i,\TT}_1$, $u'$, $a^{x_i,\TT}_0$ & where  $i \in [n]$; \\
$b^{x_i,\FF}_0$: $a^{x_i,\FF}_1$, $a^{x_i,\TT}_0$, $a^{x_i,\FF}_0$ & where  $i \in [n]$; \\
$a^{x_i,\lambda}_1$: $b^{x_i,\lambda}_1$, 
$N_{G[F]}(a^{x_i,\lambda})$,
$b^{x_i,\lambda}_0 \quad$ & where  $i \in [n]$; \\
$b^{x_i,\lambda}_1$: $a^{x_i,\lambda}_0$, $a^{x_i,\lambda}_1$ & where  $i \in [n]$;\\
$a^{c_j}$: $b_0^{c_j,3}$  & where  $j \in [m]$;\\
$b^{x_i}$: $a^{x_i,\FF}_0$ & where  $i \in [n]$.
\end{longtable}

\begin{figure}[!th]
\makebox[\textwidth][c]{
\includegraphics[scale=1]{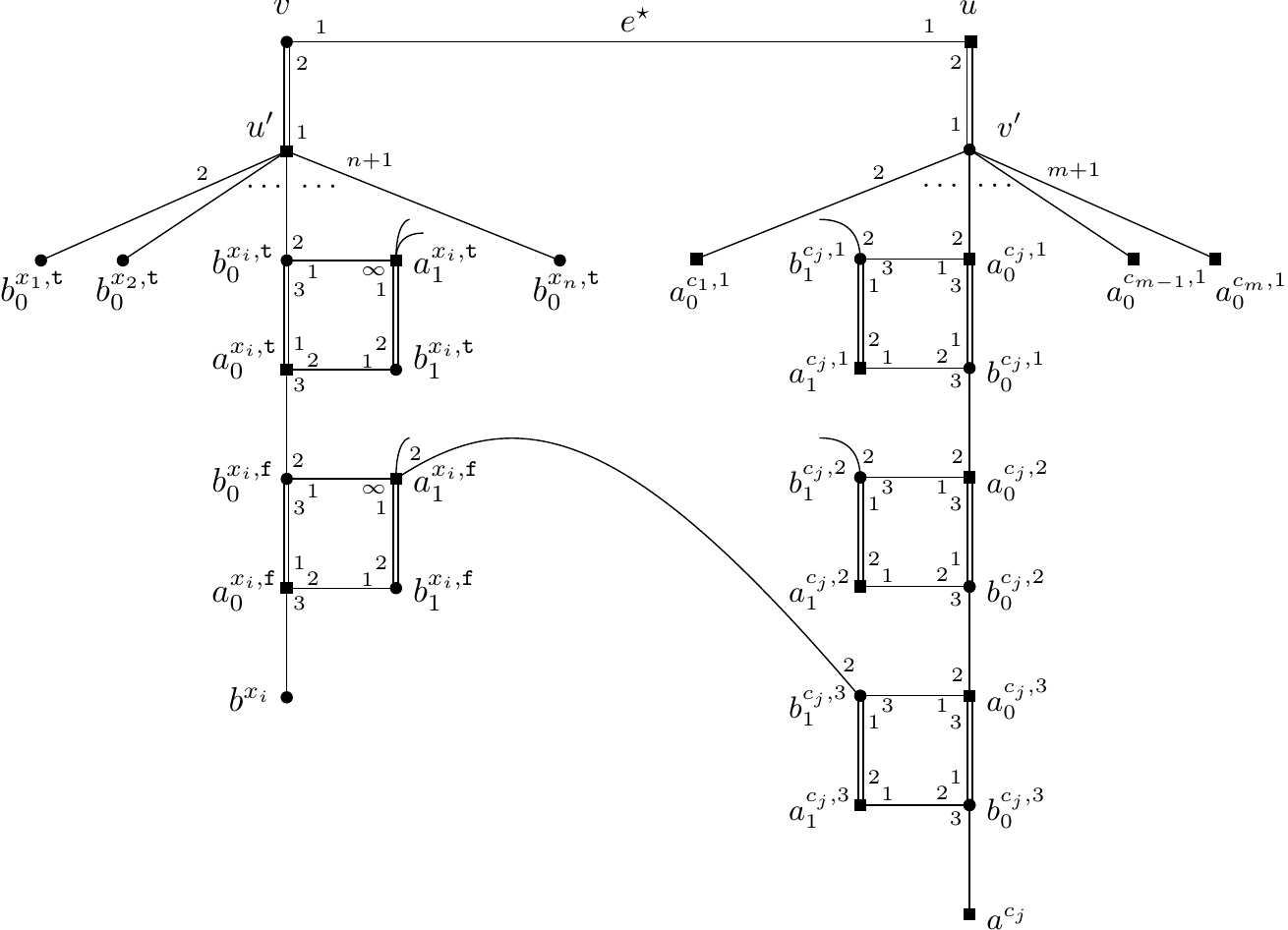}
}
\caption{Illustration of the reduction in the proof of Theorem~\ref{thm:k=1}. Edges of $M_0$ are depicted as double lines. 
The figure only shows one consistency edge, corresponding to a situation where the third literal in $c_j$ is $x_i$ as a positive literal.
}
\label{fig:hardness}
\end{figure}
We finish the construction by setting the cost function as $c \equiv 1$, our budget as $k=1$, the desired utility value as $t=1$, and the utility function as 
$$ \omega(e)=
\left\{
\begin{array}{ll}
1 & \qquad \textrm{ if $e=(u,v')$,} \\
0 & \qquad \textrm{ otherwise.}
\end{array}
\right.
$$

\paragraph{Correctness.}
First, let us observe that the utility of a matching~$M$ achieves our objective value $t=1$ if and only if it contains the edge~$(u,v')$. However, if $M$ contains $(u,v')$, then the edge $e^\star=(u,v)$ must be a blocking edge. Since $e^\star$ has cost~$1$ and our budget is $k=1$, the edge $e^\star$ must be the unique blocking edge. 
Therefore, if $M$ is feasible, then $\bp(M)=\{e^\star\}$. In fact, the converse is also true: if a popular matching is blocked solely by~$e^\star$, then by Observation~\ref{obs:characterization} it must match $u$, and hence must contain the edge~$(u,v')$. Thus, $(G,\prec,\omega,c,t,k)$ is a `yes'-instance of \MUPMU{} if and only if $(G,\prec,\{e^\star\})$ is a `yes'-instance of \FUP{}.

Let us define a matching $M_0$ that contains edges $(u,v')$ and $(u',v)$ and otherwise matches each vertex $a^{\sigma}_h \in A^x \cup A^c$ for any possible values of~$\sigma$ and~$h$ to the vertex~$b^{\sigma}_h \in B^x \cup B^c$ (see again Figure~\ref{fig:hardness}). 
To see that $M_0$ is stable in $G-e^\star$, note that vertices in $A^x \cup B^c \cup \{u',v'\}$  get their top choice in $M_0$, and they cover all edges (including consistency edges) in $G-e^\star$.
Observe that $M_0$ leaves exactly the vertices in $\hat{A} \cup \hat{B}$ unmatched.

We now show that $\varphi$ is satisfiable if and only if $(G,\prec,\omega,c,t,k)$ admits a feasible popular matching.

\paragraph{Direction ``$\Rightarrow$''.}
Assume that a truth assignment $\alpha:\{x_1,\dots,x_n\} \rightarrow \{\TT,\FF \}$ satisfies $\varphi$, where $\TT$ and $\FF$ stand for \texttt{true} and \texttt{false}, respectively.
Let $\tau(c_j)$ denote some 
literal in $c_j$ that is set to \texttt{true} by~$\alpha$.
We define a matching $M$ through determining its symmetric difference with $M_0$ as
$$M \triangle M_0=\left(\bigcup_{i=1}^n C^{x_i,\alpha(x_i)}\right) \cup \left(\bigcup_{j=1}^m C^{c_j,\tau(c_j)}\right).$$

We claim that $M$ is a popular matching in $G$ with $\bp(M)=\{e^\star\}$; we prove this by checking all conditions in Observation~\ref{obs:characterization}.

Condition~(c1) clearly holds, as $e^\star=(u,v)$ blocks $M$. To see condition~(c2), we show that no other edge blocks $M$.
Note that any blocking edge must be incident 
to some vertex in $V(M \triangle M_0)$. 
However, no edge contained in a cycle~$C^{x_i,\alpha(x_i)}$ can block~$M$, since $M$ assigns vertices of~$B$ in that cycle their top choice.
Similarly, no edge in a cycle~$C^{c_j,\tau(c_j)}$ may block~$M$, since $M$ assigns vertices of~$A$ in that cycle their top choice. Neither can edges of~$P_j$ or~$Q_i$ block~$M$, due to similar reasons.
Hence, any blocking edge must be a consistency edge. 
However, both endpoints of a consistency edge~$f \in F$ are assigned their top choice in~$M_0$, so $f$ can only  block~$M$ if both of its endpoints belong to~$V(M \triangle M_0)$. 
Let $f=(a^{x_i,\lambda},b^{c_j,\ell})$ be a consistency edge; then $x_i$ is the $\ell$-th variable in $c_j$. Suppose that both endpoints of $f$ are contained in $V(M \triangle M_0)$: then $\tau(c_j)=\ell$ and $\alpha(x_i)=\lambda$. 
However, by construction of $F$, if $c_j$ contains $x_i$ as a positive literal, then  $\lambda=\texttt{f}$, but then  setting $x_i$ to $\texttt{false}$ does not yield a true literal in $c_j$, contradicting $\alpha(x_i)=\lambda=\texttt{f}$. 
Similarly, if $c_j$ contains $x_i$ as a negative literal, then $\lambda=\texttt{t}$, but then setting $x_i$ to $\texttt{true}$ does not yield a true literal in $c_j$, contradicting $\alpha(x_i)=\lambda=\texttt{t}$.
Hence, $f \notin \bp(M)$.
Thus, $M$ is stable in~$G-e^\star$.

Let us now show that condition~(c3) holds. 
Let us call an edge $(a,b)$ where both~$a$ and~$b$ prefer their partner in~$M$ to each other a $(-,-)$ edge;
recall that $(G-e^\star)_M$ is obtained by deleting all $(-,-)$ edges from $G-e^\star$.
Clearly, the even-length $M$-alternating paths leading from vertices unmatched by $M$ (and by $M_0$) to $u$ are exactly the paths $P^\triangle_j:=P_j \triangle C^{c_j,\tau(c_j)}$, $j \in [m]$,
and similarly, the even-length $M$-alternating paths leading from vertices unmatched by $M$ to $v$ are exactly the paths $Q^\triangle_i:=Q_i \triangle C^{x_i,\alpha(x_i)}$, $i \in [n]$. 
Note that any path~$P^\triangle_j$ contains a $(-,-)$ edge w.r.t.\ $M$, namely the edge connecting $a^{c_j,\tau(c_j)}_0$ to its second choice, that is, $v'$ if $\tau(c_j)=1$ and $b_0^{c_j,\tau(c_j)-1}$ otherwise. 
Similarly, $Q^\triangle_i$ contains a $(-,-)$ edge w.r.t.\ $M$, namely the edge connecting $b^{x_i,\alpha(x_i)}_0$ to its second choice, that is, $u'$ if $\alpha(x_i)=\TT$ and~$a^{x_i,\TT}_0$ otherwise. 
So neither $P^\triangle_j$ nor $Q^\triangle_i$ is a path in $(G-e^\star)_M$ for any $j$ or $i$. Hence, condition (c3/i) holds.

Let us now prove that there is no $M$-alternating path in $(G-e^\star)_M$ from~$u$ to~$v$. Observe that any such path must have the following properties. 
\begin{itemize}
\item It starts with a subpath of $P_j$ for some $j \in [m]$, 
\item reaches the cycle $C^{c_j,\ell}$ for some $\ell \in [3]$ through the second choice of $a^{c_j,\ell}_0$ (let us call this edge $e^{c_j}$), 
\item then after traversing edges of this cycle goes through 
some consistency edge $(b^{c_j,\ell}_1,a^{x_i,\lambda})$ for some $\lambda \in \{\texttt{f,t}\}$ and $i \in [n]$, 
\item traverses edges of the cycle $C^{x_i,\lambda}$ and leaves it via the second choice of $b^{x_i,\lambda}_0$ (let us call this edge $e^{x_i}$). 
\end{itemize}
If $\ell=\tau(c_j)$, then 
$e^{c_j}$ is a $(-,-)$ edge with respect to~$M$. Similarly, if $\lambda = \alpha(x_i)$, then $e^{x_i}$ is a $(-,-)$ edge. 
If neither of these conditions hold, then $b_1^{c_j,\ell}$ and~$a_1^{x_i,\lambda}$ both get their top choice in $M$, implying that $ (b^{c_j,\ell}_1,a_1^{x_i,\lambda})$ is a $(-,-)$ edge. Hence, any $M$-alternating path from $u$ to $v$ contains a $(-,-)$ edge, and therefore condition~(c3/ii) holds as well, proving this direction of our reduction.

\paragraph{Direction ``$\Leftarrow$''.} Let us now suppose that $M$ is a popular matching in $G$ blocked only by edge $e^\star=(u,v)$. 
By Observation~\ref{obs:characterization}, $M$ is stable in $G-e^\star$, and so $(u,v')$ and $(v,u')$ are both in $M$. Also, since all stable matchings in $G-e^\star$ leave the same vertices unmatched, we know from the stability of $M_0$ in $G-e^\star$ that the set of unmatched vertices in $M$ is $\hat{A} \cup \hat{B}$.
Using this, a simple reckoning of the structure of~$G$ implies that $M$ contains exactly two edges of each 4-cycle, and does not contain any consistency edge. Note that for each $j \in [m]$ the cycle~$C^{c_j,\ell}$ must be contained in $M \triangle M_0$ for some $\ell \in [3]$, as otherwise $P_j$ would violate condition~(c3/i); we set $\tau(c_j)=\ell$ for such an $\ell$. 
Analogously, for each $i \in [n]$ the cycle~$C^{x_i,\lambda}$ must be contained in $M \triangle M_0$ for some $\lambda \in \{\TT,\FF\}$, as otherwise $Q_i$ would violate condition~(c3/i); 
we set $\alpha(x_i)=\lambda$.

We claim that the truth assignment $\alpha$ satisfies $\varphi$. 
To prove this, let us consider some~$j \in [m]$.
By the definition of $\tau$, we know that $M(b^{c_j,\tau(c_j)}_1)=a^{c_j,\tau(c_j)}_0$, which is the worst choice of $b^{c_j,\tau(c_j)}_1$.
Since the consistency edge $f$ incident to~$b^{c_j,\tau(c_j)}_1$ cannot block $M$ in $G-e^\star$ 
by condition~(c2), we obtain that $f$'s other endpoint, 
let us denote it by $a^{x_i,\lambda}_1$, must be matched to its top choice by $M$. Therefore, we get $\lambda \neq \alpha(x_i)$. 
By our definition of the consistency edges, this implies that the~\mbox{$\tau(c_j)$-th} literal in $c_j$ is set to \texttt{true} by $\alpha$,
finishing our proof.
\qed
\end{proof}

Theorem~\ref{thm:k=1-bounded} is the strengthening of Theorem~\ref{thm:k=1} for the case when the input graph has maximum degree~$3$, and preferences are single-peaked and single-crossing. 
To obtain this generalization, we reduce from a variant of 3-SAT where each variable occurs at most three times; additionally, we use a well-known technique where vertices~$a$ with $\delta(a) > 3$ are replaced by a path. Despite the conceptual simplicity of this modification, maintaining the crucial properties of the reduction while also ensuring that the constructed preference system is single-peaked and single-crossing
is a delicate task.

\repeattheorem{repthm_k=1bounded}

\begin{proof}
We give a reduction of the variant of 3-SAT where each clause has three literals and each variable occurs at most three times. 
We are going to modify the reduction presented in the proof of Theorem~\ref{thm:k=1}, re-using all definitions there.
See Figure~\ref{fig:hardness-d3} for an illustration.

\begin{figure}[th]
\makebox[\textwidth][c]{
\includegraphics[width=1.2\textwidth]{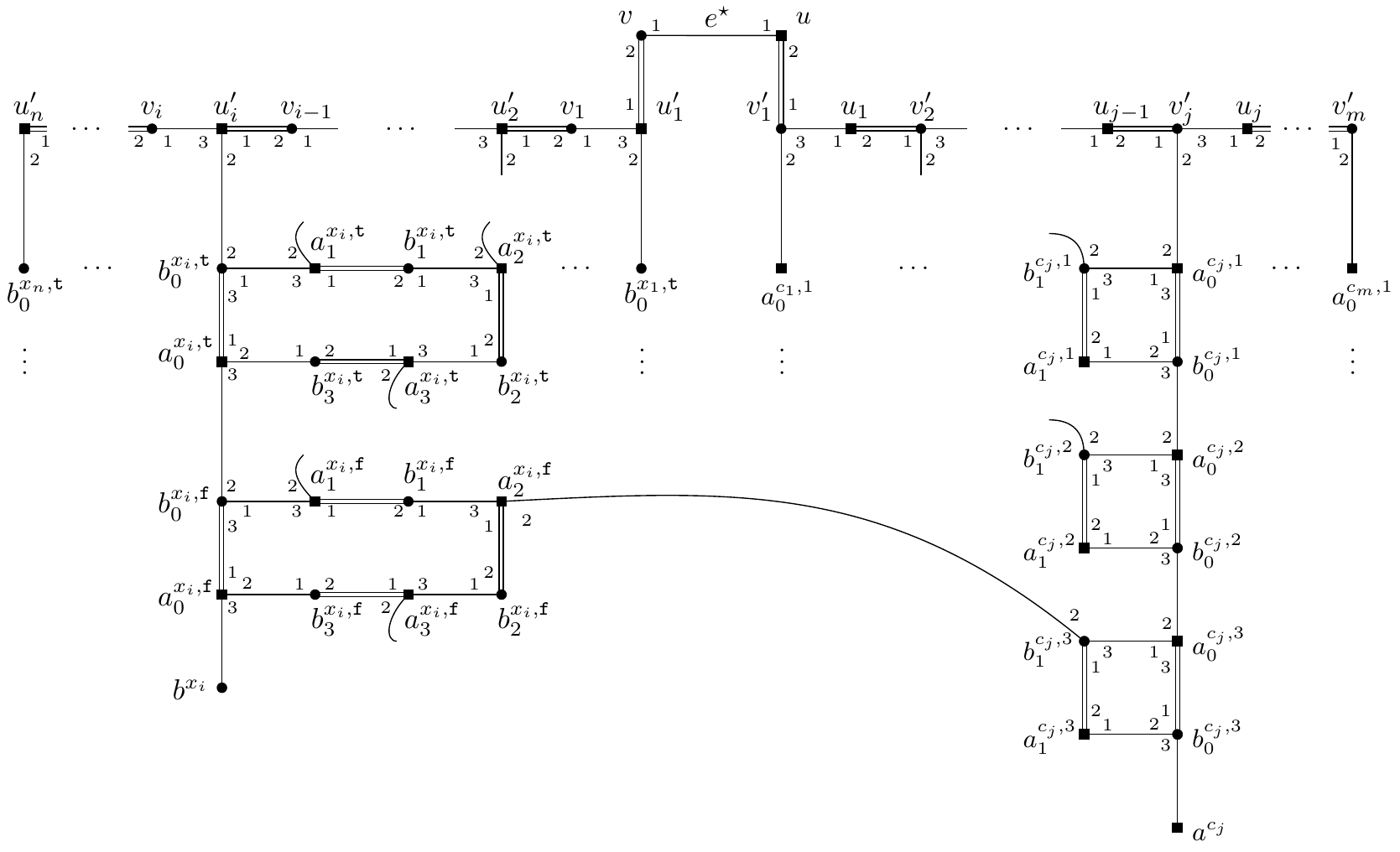}
}
\caption{Illustration of the reduction in the proof of Theorem~\ref{thm:k=1-bounded}. Edges of $\MM_0$ are depicted as double lines. 
The figure only shows one consistency edge, corresponding to a situation where the third literal in $c_j$ is $x_i$ as a positive literal.
}
\label{fig:hardness-d3}
\end{figure}

\paragraph{Construction.}
First, for each $i \in [n]$ and $\lambda \in \{\TT,\FF\}$
we replace $C^{x_i,\lambda}$ with an 8-cycle $\CC^{x_i,\lambda}$ containing vertices $a^{x_i,\lambda}_h$ and $b^{x_i,\lambda}_h$ for $h \in \{0,1,2,3\}$ (in the usual order). We replace the set $F$ with a new set $\F'$ of consistency edges as follows: 
for each $j \in [m]$ and $\ell \in [3]$, if the $\ell$-th literal of $c_j$ is the $h$-th occurrence of variable~$x_i$ for some $i \in [n]$ and $h \in [3]$, 
then we add the edge~$(a^{x_i,\lambda}_h,b^{c_j,\ell}_1)$ to~$\F'$, 
where $\lambda=\FF$ if the $h$-th occurrence of~$x_i$ is as a positive literal (the $\ell$-th literal in~$c_j$) and $\lambda=\TT$ otherwise.

Next, we replace $v'$ by a path $(v'_1, u_1, \dots, v'_{m-1},u_{m-1},v'_m)$,
replacing each edge $(v',a^{c_j,1}_0)$ with the edge $(v'_j,a^{c_j,1}_0)$.  
Analogously, we replace $u'$ by a path $(u'_1, v_1, \dots, u'_{n-1},v_{n-1},u'_n)$, 
replacing each edge $(u',b^{x_i,\TT}_0)$ with the edge $(u'_i,b^{x_i,\TT}_0)$.

The preferences of the vertices are as follows; we omit those vertices whose preferences are the same as defined in the proof of Theorem~\ref{thm:k=1}.

\begin{longtable}{ll}
$u$: $v,v'_1$; & \\
$v$: $u,u'_1$; & \\
$u'_i$: $v_{i-1}$, $b^{x_i,\TT}_0$, $v_i$ & where $i \in [n-1]$ and $v_0=v$; \\
$u'_n$: $v_{n-1}$, $b^{x_n,\TT}_0$ & \\
$v'_j$: $u_{j-1}$, $a^{c_j,1}_0$, $u_j$ & where $j \in [m-1]$ and $u_0=u$; \\
$v'_m$: $u_{m-1}$, $a^{c_m,1}_0$ & \\
$v_i$: $u'_i$, $u'_{i+1}$ & where $i \in [n-1]$; \\
$u_j$: $v'_j$, $v'_{j+1}$ & where  $j \in [m-1]$; \\[1pt]
$a^{c_j,1}_0$: $b^{c_j,1}_1$, $v'_j$, $b^{c_j,1}_0$ & where $j \in [m]$; \\
$b^{c_j,\ell}_1$: $a^{c_j,\ell}_1$, $a^{x_i,\lambda}_h$, $a^{c_j,\ell}_0$ & where $j \in [m]$, $\ell \in [3]$ and 
$(b^{c_j,\ell}_1,a^{x_i,\lambda}_h) \in \F'$; \\
$a^{x_i,\TT}_0$: $b^{x_i,\TT}_0$, $b^{x_i,\TT}_3$, $b^{x_i,\FF}_0$ & where $i \in [n]$; \\
$a^{x_i,\FF}_0$: $b^{x_i,\FF}_0$, $b^{x_i,\FF}_3$, $b^{x_i}$ & where $i \in [n]$; \\
$a^{x_i,\lambda}_h$: $b^{x_i,\lambda}_h$, $b^{c_j,\ell}_1$, $b^{x_i,\lambda}_{h-1}$ & where $i \in [n]$, $h \in [3]$ and $(b^{c_j,\ell}_1,a^{x_i,\lambda}_h) \in \F'$; \\
$a^{x_i,\lambda}_h$: $b^{x_i,\lambda}_h$,  $b^{x_i,\lambda}_{h-1}$ & where $i \in [n]$, $h \in [3]$ and $a^{x_i,\lambda}_h \notin V(\F')$; \\
$b^{x_i,\TT}_0$: $a^{x_i,\TT}_1$, $u'_i$, $a^{x_i,\TT}_0$ & where $i \in [n]$; \\
$b^{x_i,\lambda}_h$: $a^{x_i,\lambda}_{h+1 \! \! \mod 4}$, $a^{x_i,\lambda}_h \qquad$ & where $i \in [n]$ and $h \in [3]$.
\end{longtable}

Let $(\GG,\prec)$
be the preference system defined this way. Note that $\Delta_{\GG}=3$. 
To create an instance of \MUPMU{},
we set all utilities as~0 except for the edge $(u,v'_1)$ whose utility is~1, we set $c \equiv 1$, and we let $t=k=1$. 

\paragraph{Correctness.}
To show that $\GG$ admits a popular matching with $\bp(M)=\{e^\star\}$ if and only if $\varphi$ is satisfiable,
one can apply essentially the same arguments used in the proof of Theorem~\ref{thm:k=1};
for this, however, we need to define a new matching~$\MM_0$ instead of~$M_0$. 
To do so, for each $i \in [n]$ and $\lambda \in  \{\TT,\FF\}$
we first add the edges $(a^{x_i,\lambda}_2,b^{x_i,\lambda}_2)$ and $(a^{x_i,\lambda}_3,b^{x_i,\lambda}_3)$ to $M_0$.
Second, we replace the edge $(v,u')$ with  edges~$(v,u'_1)$ and~$(v_i,u'_{i+1})$ for each $i \in [n-1]$. 
Similarly, we replace the edge $(u,v')$ with edges~$(u,v'_1)$ and~$(u_j,v'_{j+1})$ for each~$j \in [m-1]$. 
Let $\MM_0$ denote the matching obtained this way. 
Note that $\MM_0$ is stable in~$\GG-e^\star$.

Now, the reasoning used in the proof of Theorem~\ref{thm:k=1} can be applied, with~$\MM_0$ taking the place of $M_0$, and 
with some trivial modifications where necessary.
Assuming that the preference system has the desired properties, 
this shows that \FUP{} is \NP-hard under the conditions stated in the theorem. 

To prove the result for \MUPMU{}, 
observe that $M$ is a popular matching in $\GG$ with~$\bp(M)=\{e^\star\}$ if and only if $M$ is a popular matching in $\GG$ whose utility is at least~$t=1$ and whose blocking edges have cost at most $k=1$ (here we use Observation~\ref{obs:characterization} again).

It remains to prove that the constructed preference system $(\GG,\prec)$ has the required properties.
Let $(\AA,\BB)$ denote the (unique) bipartition of $\GG$ where $\AA$ contains $u$ and $\BB$ contains~$v$.
 
  
  \paragraph{Single-peaked property.}
 To show single-peakedness, we define two lists
 for each $j \in [m]$ and also for each $i \in [n]$:
 \begin{eqnarray*}
 A^c_j&=&(a^{c_j},a_0^{c_j,3},a_1^{c_j,3},a_0^{c_j,2},a_1^{c_j,2},a_0^{c_j,1},a_1^{c_j,1}), \\
 B^c_j&=&(b_0^{c_j,3},b_1^{c_j,3},b_0^{c_j,2},b_1^{c_j,2},b_0^{c_j,1},b_1^{c_j,1}), \\
 A^x_i&=&(a_3^{x_i,\TT},a_2^{x_i,\TT},a_1^{x_i,\TT},a_0^{x_i,\TT},a_3^{x_i,\FF},a_2^{x_i,\FF},a_1^{x_i,\FF},a_0^{x_i,\FF}), \\
 B^x_i&=&(b_3^{x_i,\TT},b_2^{x_i,\TT},b_1^{x_i,\TT},b_0^{x_i,\TT},b_3^{x_i,\FF},b_2^{x_i,\FF},b_1^{x_i,\FF},b_0^{x_i,\FF},b^{x_i}). 
 \end{eqnarray*}
 Then we can define an axis for each of $\AA$ and $\BB$ as follows:
 \begin{equation}
 \label{axis-1}
    u_{m-1},\dots,u_1,A^c_1,\dots,A^c_m,u,u'_1,\dots,u'_n,A^x_1,\dots,A^x_n; \tag{Axis $\AA$} \notag
 \end{equation}
 \begin{equation}
 \label{axis-2}
     B^c_1,\dots,B^c_m, v'_m, \dots, v'_1,v,B^x_1,\dots,B^x_n,v_1,\dots,v_{m-1}.  \tag{Axis $\BB$} \notag
 \end{equation}
 It is straightforward to check that each vertex has single-peaked preferences with respect to the axis containing its neighbors. Note that the preferences of a degree-2 vertex are always single-peaked, while the preferences of a vertex~$a$ with $\delta(a) = 3$ are single-peaked w.r.t.\ a given axis exactly if the least-preferred neighbor of~$a$ does not lie between its two other neighbors on the axis.
 
 \paragraph{Single-crossing property.}
  To show that the preference system $(\GG,\prec)$ is single-crossing, we define a complete bipartite strict preference system $(K,\prec^K)$ compatible with~$(\GG,\prec)$ that is single-crossing; here $K$ is the complete bipartite graph whose two partitions are $\AA$ and~$\BB$.
 
 Let us define the following vertex sets and lists of vertices: 
 \begin{longtable}{llll}
 $A_h^{\ell}$ & = & $\{a_h^{c_j,\ell}:j \in [m]\} \qquad$ 
 & for each  $h \in \{0,1\}$, $\ell \in [3]$; \\
 $A_h^{\lambda}$ & = & $\{a_h^{x_i,\lambda}: i \in [n]\}$
 & for each $h \in \{0,1,2,3\}$, $\lambda \in \{\TT,\FF\}$; \\
 $B_h^{\ell}$ & = & $\{b_h^{c_j,\ell}:j \in [m]\}$
 & for each $h \in \{0,1\}$, $\ell \in [3]$; \\
 $B_h^{\lambda}$ & = & $\{b_h^{x_i,\lambda}: i \in [n]\}$
 & for each  $h \in \{0,1,2,3\}$, $\lambda \in \{\TT,\FF\}$; \\
 $U'$ & = & $(u'_1, \dots, u'_n)$; \\
 $V'$ & = & $(v'_1, \dots, v'_m)$; \\
 $U_{<j}$ & = & $(u_1, \dots, u_{j-1})$
 & for each $j \in [m]$; \\[2pt]
 $U_{\geq j}$ & = & $(u_j, \dots, u_{m-1})$ 
 & for each $j \in [m-1]$; \\[2pt] 
 $V_{<i}$ & = & $(v_1, \dots, v_{i-1})$
 & for each $i \in [n]$; \\[2pt]
 $V_{\geq i}$ & = & $(v_i, \dots, v_{n-1})$ 
 & for each $i \in [n-1]$.
 \end{longtable}
 
Next, we define the preferences in~$(K,\prec^K)$. We deal with the two sides separately, so first we define~$\prec^K_b$ for each vertex~$b \in \BB$. 
Each such~$\prec^K_b$
will be one of $m+2$ 
complete total orders over $\AA$: these are $\pi_0^j$ for each $j \in [m]$, $\pi_1$ and $\pi_2$, as they are defined below. We set $\pi_0=\pi^m_0$, and for any set appearing in the definition of these orders, we fix one arbitrary ordering (used in all of the lists below). 
To help the reader, 
at each row we highlighted in bold those vertex sets that have to be ``moved'' in order to obtain the next row.

\smallskip
\begin{center}
\begin{tabular}{ll}
$\pi_0^j$: & 
$u, U_{<j},A_0^1,A_1^1,A_0^2,A_1^2,A_0^3,A_1^3,\hat{A},\bm{U_{\geq j}},A^\FF_3,A^\FF_2,A^\FF_1,A^\TT_3,A^\TT_2,A^\TT_1,U',A^\TT_0,A^\FF_0;$ \\
$\pi_0$: & 
$u, U_{<m},\bm{A_0^1},A_1^1,\bm{A_0^2},A_1^2,\bm{A_0^3},A_1^3,\bm{\hat{A}},A^\FF_3,A^\FF_2,A^\FF_1,A^\TT_3,A^\TT_2,A^\TT_1,U',A^\TT_0,A^\FF_0;$ \\
$\pi_1$: & 
$u, U_{<m},A_1^1,A_1^2,A_1^3,A^\FF_3,A^\FF_2,A^\FF_1,A^\TT_3,A^\TT_2,A^\TT_1,U',\bm{A^\TT_0},\bm{A^\FF_0},A_0^1,A_0^2,A_0^3,\hat{A};$ \\
$\pi_2$: & 
$u, U_{<m},A_1^1,A_1^2,A_1^3,A^\TT_0,A^\FF_0,A^\FF_3,A^\FF_2,A^\FF_1,A^\TT_3,A^\TT_2,A^\TT_1,U',A_0^1,A_0^2,A_0^3,\hat{A}.$ 
\end{tabular}
\end{center}

The following claims follow directly from the definitions: 
\begin{itemize}
    \item the preference list of $v'_j$ in $\GG$ is a restriction of~$\pi^j_0$ for any $j \in [m]$;
    \item the preference list of $b_0^{c_j,\ell}$ in $\GG$ is a restriction of~$\pi_0$ for any $j \in [m],\ell \in [3]$;
    \item the preference list of $b_3^{x_i,\lambda}$ in $\GG$ is a restriction of~$\pi_2$ for any $i \in [n], \lambda \in \{\TT,\FF\}$;
    \item the preference list of any other vertex of $\BB$ in $\GG$ is a restriction of $\pi_1$.
\end{itemize}
This means that each preference relation $\prec_b$, $b \in \BB$, is compatible with one of the preference relations in $\Pi=\{\pi^j_0 \mid j \in [m]\} \cup \{\pi_1,\pi_2\}$. 
Hence, to prove that $(\GG,\prec)$ is single-crossing, it suffices to prove that $(K,\prec^K)$ is single-crossing.\footnote{Of course, we have so far only looked at the preferences of vertices in $\BB$. The same has to be argued separately for vertices of $\AA$; this will be done later in the proof.}

Instead of an ordering of $\BB$ with respect to which we could  prove the single-crossing property of $(K,\prec^K)$ for vertices of $\BB$, it clearly suffices to order the elements of $\Pi$. 
Thus, we provide an ordering $\triangleright$ over $\Pi$ as $\pi_0^1 \triangleright \dots \triangleright \pi_0^m=\pi_0 \triangleright \pi_1 \triangleright \pi_2$. 
It is now straightforward to verify that for any two distinct vertices~$a',a'' \in \AA$,  all preference relations in $\Pi^{a' \prec a''}=\{\pi \in \Pi \mid \textrm{$a''$ precedes $a'$ in $\pi$}\}$ are followed by 
all preference relations in 
$\Pi^{a'' \prec a'}=\{\pi \in \Pi \mid \textrm{$a'$ precedes $a''$ in $\pi$}\}$ 
according to~$\triangleright$, or vice versa. This proves that the single-crossing property holds for vertices of~$\BB$ in~$(K,\prec^K)$.

To deal with vertices of $\AA$ in an analogous way, we define $\prec^K_a$ for each $a \in \AA$ to be one of $n+4$
complete total orders over $\BB$: these are $\phi_0^i$ for each $i \in [n]$, $\phi_1$, $\phi_2$, $\phi_3$ and $\phi_4$,  as they are defined below. We set $\phi_0=\phi^n_0$, and for any set appearing below in these orders, we fix one arbitrary ordering (used in all of the lists below). Again, changes from one row to the next are highlighted in bold.

\smallskip
\begin{center}
\begin{tabular}{ll}
$\phi_0^i$: & 
$v,V_{<i},B_3^\TT,B_3^\FF,B_2^\TT,B_2^\FF,B_1^\TT,B_1^\FF,B_1^3,B_1^2,B_1^1,V',B_0^1,B_0^2,B_0^3,B_0^\TT,B_0^\FF,\hat{B},\bm{V_{\geq i}}$; \\
$\phi_0$: & 
$v,V_{<n},B_3^\TT,B_3^\FF,B_2^\TT,B_2^\FF,\bm{B_1^\TT},\bm{B_1^\FF},B_1^3,B_1^2,B_1^1,V',B_0^1,B_0^2,B_0^3,B_0^\TT,B_0^\FF,\hat{B}$; \\
$\phi_1$: & 
$v,V_{<n},B_3^\TT,B_3^\FF,\bm{B_2^\TT},\bm{B_2^\FF},B_1^3,B_1^2,B_1^1,V',B_0^1,B_0^2,B_0^3,B_1^\TT,B_1^\FF,B_0^\TT,B_0^\FF,\hat{B}$; \\
$\phi_2$: & 
$v,V_{<n},\bm{B_3^\TT},\bm{B_3^\FF},B_1^3,B_1^2,B_1^1,V',B_0^1,B_0^2,B_0^3,B_2^\TT,B_2^\FF,B_1^\TT,B_1^\FF,B_0^\TT,B_0^\FF,\hat{B}$; \\
$\phi_3$: & 
$v,V_{<n},B_1^3,B_1^2,B_1^1,V',\bm{B_0^1},\bm{B_0^2},\bm{B_0^3},B_3^\TT,B_3^\FF,B_2^\TT,B_2^\FF,B_1^\TT,B_1^\FF,\bm{B_0^\TT},\bm{B_0^\FF},\hat{B}$; \\
$\phi_4$: & 
$v,V_{<n},B_0^1,B_0^2,B_0^3,B_1^3,B_1^2,B_1^1,V',B_0^\TT,B_3^\TT,B_0^\FF,B_3^\FF,B_2^\TT,B_2^\FF,B_1^\TT,B_1^\FF,\hat{B}$.
\end{tabular}
\end{center}

The following claims follow directly from the definitions: 
\begin{itemize}
    \item the preference list of $u'_i$ in $\GG$ is a restriction of~$\phi^i_0$ for any $i \in [n]$;
    \item the preference list of $a_1^{x_i,\lambda}$ in $\GG$ is a restriction of~$\phi_0$ for any $i \in [n],\lambda \in \{\TT,\FF\}$;
    \item the preference list of $a_2^{x_i,\lambda}$ in $\GG$ is a restriction of~$\phi_1$ for any $i \in [n],\lambda \in \{\TT,\FF\}$;
    \item the preference list of $a_3^{x_i,\lambda}$ in $\GG$ is a restriction of~$\phi_2$ for any $i \in [n],\lambda \in \{\TT,\FF\}$;
    \item the preference list of $a_0^{x_i,\lambda}$ in $\GG$ is a restriction of~$\phi_4$ for any $i \in [n], \lambda \in \{\TT,\FF\}$;
    \item the preference list of $a_1^{c_j,\ell}$ in $\GG$ is a restriction of~$\phi_4$ for any $j \in [m], \ell \in [3]$;
    \item the preference list of any other vertex in $\GG$ is a restriction of $\phi_3$.
\end{itemize}
This means that each preference relation $\prec_a$, $a \in \AA$, is compatible with one of the preference relations in $\Phi=\{\phi^i_0 \mid i \in [n]\} \cup \{\phi_1,\phi_2, \phi_3, \phi_4 \}$. 
Hence, to prove that $(\GG,\prec)$ is single-crossing with regard to vertices of $\AA$, it suffices to prove that $(K,\prec^K)$ is single-crossing with regard to vertices of $\AA$.

Instead of an ordering of $\AA$ with respect to which we could  prove the single-crossing property of $(K,\prec^K)$ for vertices of $\AA$, it clearly suffices to order the elements of $\Phi$. 
Therefore, we provide an ordering $\triangleright$ over $\Phi$ as
$\phi_0^1 \triangleright \dots \triangleright \phi_0^n=\phi_0 \triangleright \phi_1 \triangleright \phi_2 \triangleright \phi_3 \triangleright \phi_4$.
It is now straightforward to verify that for any two distinct vertices~$b',b'' \in \BB$,  all preference relations in $\Phi^{b' \prec b''}=\{\phi \in \Phi \mid \textrm{$b''$ precedes $b'$ in $\phi$}\}$ are followed by 
all preference relations in 
$\Phi^{b'' \prec b'}=\{\phi \in \Phi \mid \textrm{$b'$ precedes $b''$ in $\phi$}\}$ 
according to~$\triangleright$, or vice versa. This proves that the single-crossing property holds for vertices of~$\AA$ in~$(K,\prec^K)$ as well, 
finishing our proof.
\qed
\end{proof}

\repeattheorem{repthm_k=1maxsize}

\begin{proof}
Again, we give a reduction from the variant of 3-SAT where each clause has three literals and each variable occurs at most three times. 
We are going to modify the reduction presented in the proof of Theorem~\ref{thm:k=1-bounded}, re-using definitions there, and placing emphasis on condition~(c3/ii) instead of condition~(c3/i) of Observation~\ref{obs:characterization}.
See Figure~\ref{fig:hardness-maxsize} for an illustration.

\begin{figure}[!th]
\makebox[\textwidth][c]{
\includegraphics[scale=1]{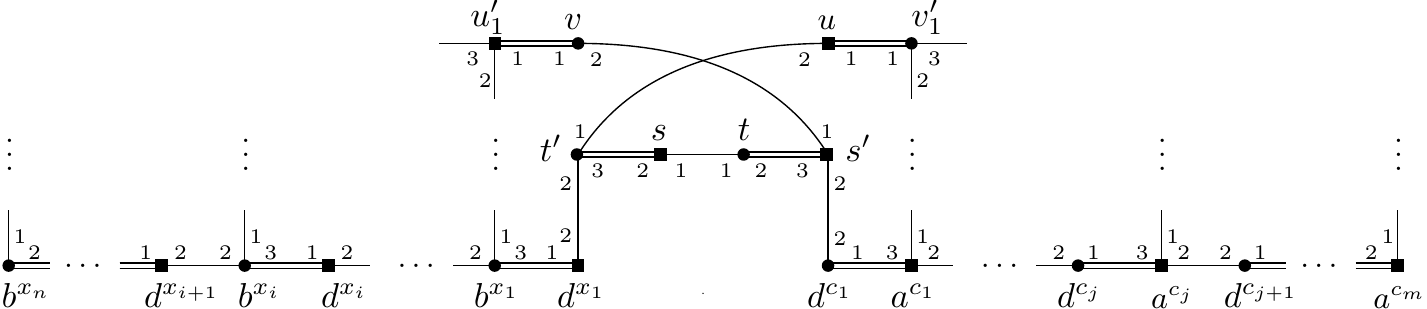}
}
\caption{Illustration of the reduction in the proof of Theorem~\ref{thm:k=1-maxsize-bounded}, depicting the part of graph $\hat{G}$ where it differs from graph $\GG$ (shown in Figure~\ref{fig:hardness-d3}). 
Double lines denote the matching~$\hat{M}_0$.
}
\label{fig:hardness-maxsize}
\end{figure}

\paragraph{Construction.}
We start with 
adding a  set~$D^X \cup D^C \cup \{s,s',t,t'\}$ of new vertices to $\GG$
where $D^X=\{d^{x_i}:i \in [n]\}$
and $D^C=\{d^{c_j}:j \in [m]\}$. Among the vertices of $\GG$, only those in $\hat{A} \cup \hat{B} \cup \{u,v\}$ will be adjacent to newly introduced vertices; 
the preferences of all other vertices remain unchanged. 
The preferences of the new vertices and of those that are adjacent to them are as follows.
\smallskip


\begin{tabular}{ll}
$s$: $t, t'$; \\$t$: $s, s'$; \\
$s'$: $v, d^{c_1}, t$; \\
$t'$: $u, d^{x_1}, s$; \\
$u$: $v'_1,t'$;\\
$v$: $u'_1,s'$; \\
$d^{x_1}$: $b^{x_1}, t'$; \\
$d^{x_i}$: $b^{x_i},  b^{x_{i-1}}$  & where $i \in \{2, \dots, n\}$; \\
$b^{x_i}$: $a_0^{x_i,\FF}, d^{x_{i+1}}, d^{x_i} \qquad \qquad\qquad \qquad$ & where $i \in [n-1]$; \\
$b^{x_n}$: $a_0^{x_n,\FF}, d^{x_n}$; \\
$d^{c_1}$: $a^{c_1}, s'$; \\
$d^{c_j}$: $a^{c_j}, a^{c_{j-1}}$ & where $j \in \{2, \dots, m\}$; \\
$a^{c_j}$: $b_0^{c_j,3}, d^{c_{j+1}}, d^{c_j}$ & where $j \in [m-1]$; \\
$a^{c_m}$: $b_0^{c_m,3}, d^{c_m}$. 
\end{tabular}
\smallskip

As promised in the theorem, we set both the cost function and the utility function to be uniformly~$1$, and we set $t=|M_s|+1$ and~$k=1$. 

\paragraph{Correctness.}
Recall the matching~$\MM_0$ defined for the graph~$\GG$ (see Figure~\ref{fig:hardness-d3}); observe that~$\MM_0$ is a matching in~$\hat{G}$ as well. We define a matching~$M_s$ by adding the edges~$\{(d^{x_i},b^{x_i}):i \in [n]\} \cup \{(d^{c_j},a^{c_j}):j \in [m]\}  \cup \{(s,t)\}$ to~$\MM_0$. 
It is straightforward to verify that $M_s$ is stable in~$\hat{G}$. Note that the only vertices left unmatched  by $M_s$ are $s'$ and $t'$. 
Therefore a matching is feasible if and only if it is complete and has at most one blocking edge. 

We are going to show that there is a feasible and popular matching in~$\hat{G}$ if and only if the input formula~$\varphi$ is satisfiable.

\paragraph{Direction ``$\Rightarrow''$.}
Suppose that $M$ is a feasible popular matching in~$\hat{G}$. 
We claim that $M(s')=t$ and $M(t')=s$.

First, if $M(t')=u$, then $(u,v'_1)$ blocks~$M$. 
Second, if $M(t')=d^{x_1}$, then at least one edge incident to some vertex of $D^X$ must block~$M$: this follows from an observation that $M$ cannot connect a vertex $b^{x_i}$ to its top choice (i.e., $a_0^{x_i,\FF}$). To see this, suppose that $(a_0^{x_i,\FF},b^{x_i}) \in M$.
This means that $(a_0^{x_i,\FF},b_3^{x_i,\FF}) \in \bp(M)$, so no other edge in~$\CC^{x_i,\FF}$ can block~$M$; 
however, this leads to \mbox{$(b_h^{x_i,\FF},a_h^{x_i,\FF}) \in M$} for~$h=3,2,1$, implying also $(b_0^{x_i,\FF},a_0^{x_i,\TT}) \in M$. 
This in turn yields that the edge~$(a_0^{x_i,\TT},b_3^{x_i,\TT})$ blocks~$M$, a contradiction. 
Hence, if $M(t') \neq s$, then at least one edge in~$\bp(M)$ is incident to a vertex of~$D^X \cup \{u\}$. 

Analogously, the same arguments yield that if $M(s') \neq t$, then at least one edge in $\bp(M)$ is incident to a vertex of $D^C \cup \{v\}$. Hence, by $|\bp(M)|\leq 1$, and since no vertex of $D^X \cup \{u\}$ is adjacent to a vertex of $D^C \cup \{v\}$, at least one of $(s',t)$ and $(s,t')$ is in $M$. 
This implies $\bp(M)=\{(s,t)\}$, and consequently, 
$\{(s,t'),(s',t)\} \subseteq M$
because $M$ is complete.

Since $M$ is complete,
 $d^{x_1}$ must be matched in~$M$, yielding $(d^{x_1},b^{x_1}) \in M$. Applying the same argument repeatedly, we get $(d^{x_i},b^{x_i}) \in M$  for $i=2,\dots, n$ as well. Similarly, we get $(d^{c_j},a^{c_j}) \in M$ for each $j \in [m]$, 
and also $(u'_i,v_{i-1}) \in M$ for each $i \in [n]$ 
and $(v'_j,u_{j-1}) \in M$ for each $j \in [m]$ 
where $v_0=v$ and $u_0=u$.  

Let us define a matching~$\hat{M}_0=M_s \setminus \{(s,t)\} \cup \{(s,t'),(s',t)\}$; see again Figure~\ref{fig:hardness-maxsize}. By the observations of the previous paragraph and using that $M$ is complete, it follows that $M \triangle \hat{M}_0$ is the union of cycles of the form~$\CC^{x_i,\lambda}$ and~$C^{c_j,h}$. 
Moreover, we claim that for each $i \in [n]$ there exists some $\lambda \in \{\TT,\FF\}$ such that 
$\CC^{x_i,\lambda} \subseteq M \triangle \hat{M}_0$, and that for each $j \in [m]$ there exists some $h \in [3]$
such that 
$C^{c_j,h} \subseteq M \triangle \hat{M}_0$.
To see this for some $i \in [n]$, consider the path 
\begin{multline}
\notag
\hat{P}_i=s,t',d^{x_1},b^{x_1},\dots, d^{x_i},b^{x_i},a_0^{x_i,\FF},
b_0^{x_i,\FF},a_0^{x_i,\TT}, b_0^{x_i,\TT},
\\
u'_i, v_{i-1, }, u'_{i-1},\dots, v_1,u'_1,v, s',t.
\end{multline}
Note that there are no $(-,-)$ edges on $\hat{P}_i$ with respect to the matching~$\hat{M}_0$. By Observation~\ref{obs:characterization}, this path cannot be present in $(\hat{G}-(s,t))_M$, implying that $M$ cannot contain all edges of $\hat{M}_0 \cap \hat{P}_i$. 
Thus, indeed there exists some $\lambda \in \{\TT,\FF\}$ for which $\CC^{x_i,\lambda} \subseteq M \triangle \hat{M}_0$; we define $\alpha(x_i)$ to be such a value of $\lambda$. 

Arguing analogously about the path 
\begin{multline}
\notag
\hat{Q}_j=t,s',d^{c_1},a^{c_1},\dots, d^{c_j},a^{c_j},b_0^{c_j,3},
a_0^{c_j,3},b_0^{c_j,2},
a_0^{c_j,2},b_0^{c_j,1},
a_0^{c_j,1},\\
v'_j, u_{j-1, }, v'_{j-1},\dots, u_1,v'_1,u, t',s,
\end{multline}
we get that for each $j \in [m]$ there exists some $h \in [3]$ for which $C^{c_j,h} \subseteq M \triangle \hat{M}_0$; we define $\tau(c_j)$ to be such a value  $h$.

Based on the fact that no consistency edge can block~$M$, one can argue in the same way as in the proof of Theorems~\ref{thm:k=1} and~\ref{thm:k=1-bounded} to prove that $\alpha$ is a truth assignment satisfying the input formula~$\varphi$. 

\paragraph{Direction ``$\Leftarrow$''.}
Assume that $\alpha:\{x_1,\dots,x_n\} \rightarrow \{\TT,\FF\}$ is a truth assignment that satisfies~$\varphi$. 
Let $\tau(c_j)=h$ if the $h$-th literal in clause~$c_j$ is set to true by~$\alpha$ (any such value $h$ works). 
Then we define $M$ through determining 
its symmetric difference with $\hat{M}_0$ as
$$M \triangle \hat{M}_0=\left(\bigcup_{i=1}^n C^{x_i,\alpha(x_i)}\right) \cup \left(\bigcup_{j=1}^m C^{c_j,\tau(c_j)}\right).$$

It is easy to see that $M$ is feasible, in particular, $\bp(M)=\{(s,t)\}$.
Using Observation~\ref{obs:characterization}, we can  show that $M$ is popular as well.
It is straightforward to check that conditions (c1), (c2), and (c3/i) of Observation~\ref{obs:characterization} hold. 
To see that (c3/ii) holds as well, let $P$ be an $M$-alternating path $P$ starting with~$(s,t')$ and ending with~$(s',t)$. 
First note that if $P$ contains no consistency edge, then it must contain a $(-,-)$ edge w.r.t.\ $M$: 
either an edge connecting some $b_0^{x_i,\alpha(x_i)}$ with its second choice (let us denote this edge by $e^{x_i}$), 
or an edge connecting some $a_0^{c_j,\tau(c_j)}$ with its second choice (let us denote this edge by $e^{c_j}$). 
If $P$ does contain some consistency edge~$f$, then either $f$ itself is a $(-,-)$ edge or it must be adjacent to an edge in some cycle~$C$ in $M\triangle \hat{M}_0$. 
Thus, either $C=\CC^{x_i,\alpha(x_i)}$ for some~$i \in [n]$, or $C=C^{c_j,\tau(c_j)}$ for some $j \in [m]$.
In both cases we can identify a $(-,-)$ edge on~$P$, namely $e^{x_i}$ in the former case, and $e^{c_j}$ in the latter case. This proves that $M$ is indeed popular, and hence the reduction is correct.

It remains to prove that the constructed preference system~$(\hat{G},\hat{\prec})$ is single-peaked and single-crossing. 

\paragraph{Single-peaked property.}
Observe that all newly added vertices except for $s'$ and $t'$ have degree~$2$ in $\hat{G}$, as do $u$ and $v$ as well. 
Hence, these vertices have preferences that are trivially single-peaked with respect to any axis. To deal with $s'$, $t'$, and the vertices in $\hat{A} \cup \hat{B}$, it suffices to append the vertices $d^{x_n},d^{x_{n-1}}, \dots, d^{x_1},s,s'$ in this order
after \ref{axis-1}, and similarly, to append
$d^{c_m},d^{c_{m-1}}, \dots, d^{c_1},t,t'$ in this order
after \ref{axis-2}.
It can easily be verified that $(\hat{G},\hat{\prec})$ is single-peaked with respect to the two axis obtained this way.

\paragraph{Single-crossing property.}
To show that $(\hat{G},\hat{\prec})$ is single-crossing, recall the orders in~$\Pi$ and in~$\Phi$ as defined in the proof of Theorem~\ref{thm:k=1-bounded}.
We apply two modifications to each $\pi \in \Pi$: first, we append the vertices $d^{x_n},d^{x_{n-1}}, \dots, d^{x_1},s,s'$ to $\pi$ in this order; 
second, we fix the order of the vertices in $\hat{A}$
as $a^{c_m}, a^{c_{m-1}}, \dots, a^{c_1}$. 
Similarly, we apply two modifications to each $\phi \in \Phi$: first, we append the vertices $d^{c_m}, d^{c_{m-1}}, \dots, d^{c_1},t,t'$ to $\phi$  in this order; 
second, we fix the order of the vertices in $\hat{B}$
as $b^{x_n},b^{x_{n-1}}, \dots, b^{x_1}$. It is straightforward to check that the preference list of any vertex in $D^C \cup \{t,t'\}$ can be obtained as the restriction of any $\pi \in \Pi$, while
the preference list of any vertex in $D^X \cup \{s,s'\}$ can be obtained as the restriction of any $\phi \in \Phi$. Moreover, any vertex whose preference list in $\GG$ is the restriction of some $\pi \in \Pi$  or $\phi  \in \Phi$ has the same property in $\hat{G}$ with respect to the modified orders (as defined above). Therefore,  $(\hat{G},\hat{\prec})$ is single-crossing, proving the theorem.
\qed
\end{proof}

\begin{myremark}
\label{remark:altproof}
The \NP-hardness of finding a complete popular matching with at most one blocking edge (a weaker form of Theorem~\ref{thm:k=1-maxsize-bounded} and Corollary~\ref{cor:maxsize}) 
can also be proved  
using the reduction in~\cite[Section 5.1]{CFK+22}.
It can be verified that adding a path $(s,x,y,t)$ to the graph constructed in the reduction presented in \cite[Section 5.1]{CFK+22}, with newly introduced vertices $x$ and $y$ being each other's top choice, and with $s$ ranking $x$, as well as $t$ ranking $y$ as their worst choice, the obtained graph admits a complete popular matching if and only if it admits a complete popular matching with $(x,y)$ as the unique blocking edge, which in turn happens if and only if the input instance of 3-SAT is satisfiable. This argument is due to Telikepalli Kavitha (personal communication);
the reduction we present in the proof of  Theorem~\ref{thm:k=1-maxsize-bounded} is a combination of her ideas and our reduction for Theorem~\ref{thm:k=1-bounded}.
\end{myremark}

\subsection{Missing proofs from Section~\ref{sec:masterlist}}
\label{sec:app-masterlist}

\repeattheorem{repthm_ML_W1hard}

\begin{proof}
We give a reduction from the $\mathsf{W}[1]$-hard \myproblem{Multicolored Clique} problem, parameterized by the size of the solution~\cite{FHRV09}. The input of this problem is a graph $G=(V,E)$ and an integer~$q$, 
with the vertex set of $G$ partitioned into sets~$V_1, \dots, V_q$, and the task is to decide whether $G$ contains a clique of size~$q$ containing exactly one vertex from each of the sets $V_i$. 
We let $E_{i,j}$ denote the edges of $G$ that run between $V_i$ and $V_j$ for some $1 \leq  i < j \leq q$. 

\paragraph{Construction.}
For each $i \in [q]$, we construct a vertex gadget $G'_i$ on vertex set $A_i \cup B_i \cup \{s_i,t_i\}$ where
$A_i=\{a_v: v \in V_i\}$ and
$B_i=\{b_v: v \in V_i\}$. 
For each~$v \in V_i$, 
vertices~$a_v$ and~$b_v$ are each other's top choices, 
the worst choice of~$a_v$ is~$s_i$, and 
the worst choice of~$b_v$ is~$t_i$; 
the gadget contains only these~$3|V_i|$ edges.
Similarly, for each $i,j \in [q]$ with $i <j$, we construct an edge gadget $G'_{i,j}$ on vertex set $A_{i,j} \cup B_{i,j} \cup \{s_{i,j},t_{i,j}\}$ where
$A_{i,j}=\{a_e:e \in E_{i,j}\}$ and
$B_{i,j}=\{b_e:e \in E_{i,j}\}$. 
For each $e \in E_{i,j}$, 
vertices $a_e$ and $b_e$ are each other's top choices, 
the worst choice of $a_e$ is $s_{i,j}$, and 
the worst choice of $b_e$ is $t_{i,j}$;
the gadget contains only these~$3|E_{i,j}|$ edges. We also add vertices~$s_0$ and~$t_0$, connected with each other.

We next create edges that connect our gadgets along a path. 
For this, we need an ordering over our gadgets, so
let $\mu$ be any fixed bijection from $[q+\binom{q}{2}]$ 
to \mbox{$[q] \cup \{(i,j) \in [q] \times [q]: i<j \}$}
satisfying the property that all vertex gadgets precede all edge gadgets (formally, this amounts to  $\mu^{-1}(i)<\mu^{-1}(i',j')$ for any~$i,i',j' \in [q]$ with~$i'<j'$). For simplicity, we also assume $\mu(1)=1$.
We can now add the edge set

\begin{center}
$E_{TS}= \{(t_0,s_1)\} \cup\{(t_{\mu(h)},s_{\mu(h)+1}): h \in [q+\binom{q}{2}-1 ]\}.$
\end{center}

We further create a set~$F$ of \emph{consistency edges} as follows. For each edge~$e \in E$ connecting vertices~$x \in V_i$ and~$y \in V_j$ for some~$i$ and~$j$, we connect vertex~$b_e$ with all vertices in $\{a_v : v \in V_i \setminus \{x\} \} \cup \{a_v : v \in V_j \setminus \{y\} \}$. 

\begin{figure}[th]
\makebox[\textwidth][c]{
\includegraphics[scale=0.8]{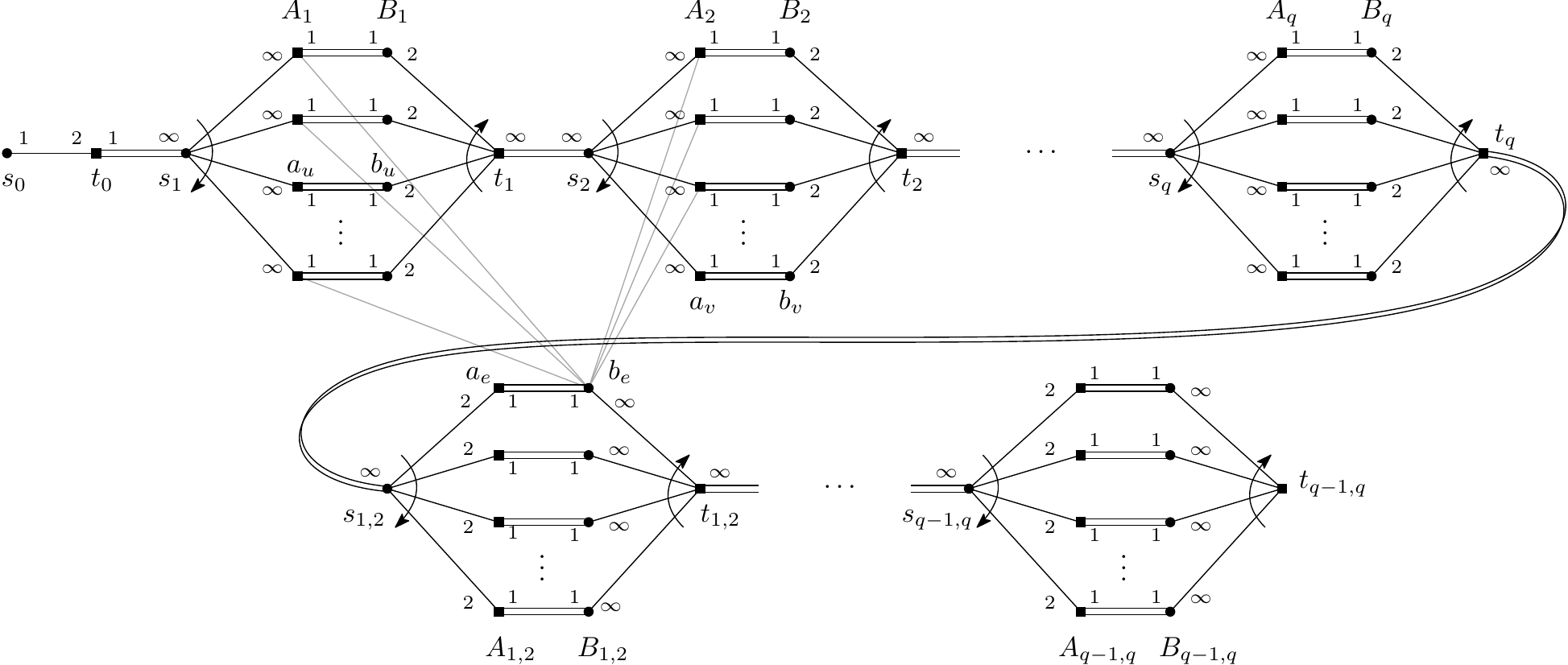}
}
\caption{Illustration of the reduction in the proof of Theorem~\ref{thm:masterlist-W1-hard}. Vertex gadgets are shown in the upper row, with edge gadgets below them. We assume $\mu(i)=i$ for each~$i \in [q]$, as well as $\mu(q+1)=(1,2)$ and $\mu(q+\binom{q}{2})=(q-1,q)$. 
Double lines denote edges of~$M_s$. Among consistency edges only those incident to $b_e$ are shown, in grey.
The figure assumes $e=(u,v) \in E_{1,2}$.
}
\label{fig:2ML-hardness}
\end{figure}

Let $G'$ denote the resulting graph; see Figure~\ref{fig:2ML-hardness} for an illustration.
We set 
\[
A_V = \bigcup_{i=1}^q A_i, \quad
A_E = \bigcup_{1 \leq i<j \leq q} A_{i,j},  \quad
S=\{s_i:i \in [q]\} \cup \{s_{i,j} : 1 \leq i < j \leq q \} 
\]
and we define the sets $B_V$, $B_E$, and $T$ analogously. 
Then $A=A_V \cup A_E \cup S \cup \{s_0\}$ and $B=B_V \cup B_E \cup T \cup \{t_0\}$ are the two partitions of $G'$.
We proceed with providing a master list over vertices in $A$ and in $B$, denoted by $\mathcal{L}_A$ and $\mathcal{L}_B$, respectively.
To this end, we fix an arbitrary ordering~$\pi$ over $V \cup E$. Then we write $\overrightarrow{A_V}$ to denote the ordering of $A_V$ in which $a_v$ precedes $a_{v'}$ if and only if $v$ precedes $v'$ according to~$\pi$. 
By contrast, let $\overleftarrow{B_V}$ to denote 
the ordering of $B_V$ in which $b_v$ precedes $b_{v'}$ for some $v' \neq v$ if and only if $v$ does \emph{not} precede $v'$ according to~$\pi$. 
We define $\overrightarrow{A_E}$ 
and~$\overleftarrow{B_E}$ analogously. 
We are now ready to define the preferences.
$$
\begin{tabular}{ll}
$\mathcal{L}_A$: $\overrightarrow{A_E}$, $\overrightarrow{A_V}$, $S$, $s_0$;\\
$\mathcal{L}_B$: $\overleftarrow{B_V}$, $\overleftarrow{B_E}$, $T$, $t_0$.
\end{tabular}
$$

Observe that $M_s=\{(a_v,b_v):v \in V\} \cup \{(a_e,b_e):e \in E\} \cup E_{TS}$ is the  unique stable matching in $G'$.

All edges have cost~$1$, and our budget is $k=q+\binom{q}{2}$. %
Regarding the utility function $\omega$ and the objective value~$t$, we define two equivalent variants, as required.  To show hardness for variant~(a) we set $\omega$ such that the edge~$(s_0,t_0)$ has utility~1, while all other edges have utility~0, and we set $t=1$. 
Note that $\omega(M_s)=0$ indeed holds in this case.
For variant~(b), we set $\omega\equiv 1$ and $t=|V(G')|/2$;
note that in this case~$t=|M_s|+1$. 
Observe that the presented reduction is
a parameterized reduction with parameter $k$ (in both variants).

We prove that there exists a feasible
popular matching in the constructed instance~$(G',\prec,\omega,c,t,k)$ if and only if $G$ contains a clique of size~$q$ containing a vertex from each set $V_i$.

\paragraph{Direction``$\Rightarrow$''.} Suppose that $M$ is a matching that is popular and feasible for~$(G',\prec,\omega,c,t,k)$. 

First, let us observe that a popular matching~$M$ cannot leave a vertex~$v$  in~$A_V \cup A_E \cup B_V \cup B_E$ unmatched: supposing otherwise we immediately get that $(v,M_s(v))$ blocks $M$, and so $v$ cannot be unmatched due to condition~(c3/i) of Observation~\ref{obs:characterization}, a contradiction.

\medskip
\noindent
\emph{Claim.} ($\blacklozenge$)
If $(s_{\mu(h)},a) \in M$ for some vertex~$a \in A_{\mu(h)}$, then $(M_s(a),t_{\mu(h)}) \in M$ as well, and moreover,  $M(s_{\mu(h+1)}) \in A_{\mu(h+1)}$ unless $h=q(q+1)/2$. 

\medskip

We prove this  first for vertex gadgets, i.e., for $h \in [q]$. 
Suppose $(s_{\mu(h)},a) \in M$ for some $a \in A_{\mu(h)}$ with $h \in[q]$. 
Let $b=M_s(a)$. Since $a$ and $b$ are each other's top choice, $(a,b) \in \bp(M)$ and so by condition~(c3/i) of Observation~\ref{obs:characterization}
we know that $b$ must be matched in~$M$. 
Since $G'_{\mu(h)}$ is a vertex gadget, $\delta_{G'}(b)=2$.
Hence, we get $(b,t_{\mu(h)}) \in M$. 
Using again condition~(c3/i), we also know that $s_{\mu(h+1)}$ must be matched in~$M$,
and since it clearly cannot be matched to~$t_{\mu(h)}$,  claim~($\blacklozenge$) follows in this case. 

Before proving claim~($\blacklozenge$) for edge gadgets, let us show first that $M$ contains no consistency edges.
By the feasibility of~$M$ we know $(s_0,t_0) \in M$, and hence $s_1$ must be matched by~$M$, as otherwise $(t_0,s_1)$ would be a blocking edge of~$M$ with an unmatched endpoint, contradicting condition~(c3/i) of Observation~\ref{obs:characterization}. Thus, $M(s_1) \in A_1$. Moreover, by our choice of~$\mu$, all vertex gadgets precede all edge gadgets, and thus repeatedly applying claim~($\blacklozenge$) for vertex gadgets  (proved in the previous paragraph)  implies that
for each~$i \in [q]$ there exists a unique vertex $v \in V_i$ such that $(s_i,a_v)$ and $(b_v,t_i)$ are both in~$M$, with~$(a_v,b_v)\in \bp(M)$; we denote this vertex~$v$ by~$v_i$. 
Assuming that some consistency edge~$(a_v,b_e) \in F$ is contained in~$M$, this would imply that $b_v$ cannot be matched to any vertex by~$M$ (since $a_v \neq a_{v_i}$ for any~$i \in [q]$ by~$M(a_v) \notin S$), a contradiction.

Suppose now that  $(s_{\mu(h)},a) \in M$ for some $a \in A_{\mu(h)}$ in an edge gadget~$G'_{\mu(h)}$.
Again, let $b=M_s(a)$. Then $b$ must be matched in~$M$ by condition~(c3/i) of Observation~\ref{obs:characterization}.
Note that $\delta_{G'-F}(b)=2$, and hence from $M \cap F = \emptyset$ we get~$(b,t_{\mu(h)}) \in M$.
Thus, if $s_{\mu(h+1)}$ exists, then it must also be matched in~$M$ by condition~(c3/i), and can only be matched to some vertex in~$A_{\mu(h+1)}$, proving claim ($\blacklozenge$).
Therefore, for any~$i,j \in [q]$ with~$ i<j$ there exists a unique edge~$e \in E_{i,j}$ such that $(s_{i,j},a_e)$ and~$(b_e,t_{i,j})$ are both in~$M$, with~$(a_e,b_e) \in \bp(M)$; let $e_{i,j}$ denote this edge~$e$.

As a consequence, each gadget contains an edge of $M_s$ that blocks~$M$. By $k=q+\binom{q}{2}$ we get that
no consistency edge can block $M$ without exceeding the budget. 
We claim that each edge $e_{i,j}$ connects~$v_i$ and~$v_j$. 
Indeed, if $e_{i,j}$ is not incident to $v_i$, then $a_{v_i}$ and $b_{e_{i,j}}$ are connected in $G'$ by a consistency edge, 
and they form a blocking edge for $M$ (as both of them are matched to their least favorite neighbor in $M$), a contradiction. By symmetry, $e_{i,j}$ is incident to~$v_j$ as well, proving our claim. Therefore, vertices~$v_1, \dots, v_q$ form a clique in~$G$, with~$v_i \in V_i$ for each~$i \in [q]$ as desired. 

\paragraph{Direction ``$\Leftarrow$''.} Suppose that vertices $v_i \in V_i$, $i \in [q]$, form a clique in~$G$; let $K=G[\{v_1,\dots,v_q]\}$. We define a matching~$M$ as follows.
\begin{eqnarray*}
M&=&\{(s_0,t_0)\} \cup \{(a_v,b_v):v \in V \setminus V(K)\} 
\cup  \{(s_i,a_{v_i}),(b_{v_i},t_i): i \in [q]\} \\ & &
\cup 
\{(a_e,b_e):e \in E \setminus E(K)\}  \cup 
\{(s_{i,j},a_{(v_i,v_j)}),(b_{(v_i,v_j)},t_{i,j}): i,j \in [q],i<j\}.
\end{eqnarray*}
Observe that $M$ is a complete matching, and the set of edges blocking $M$ is $$\bp(M)=\{(a_{v_i},b_{v_i}):i \in [q]\} \cup 
\{(a_{(v_i,v_j)},b_{(v_i,v_j)}):i,j \in [q],i<j\},$$
since all vertices in $A_V \cup B_V \cup A_E \cup B_E$ except those in $V(\bp(M))$ are matched to their top choice in $M$; note that since $K$ is a clique, the definition of $G'$ implies that no consistency edge has both of its endpoints in  $V(\bp(M))$, and therefore no consistency edge blocks $M$. This shows that the edges blocking $M$ have total cost exactly $k=|\bp(M)|$, and thus $M$ is feasible. It remains to show that $M$ is popular in $G'$. 

For the sake of contradiction, suppose that $M'$ is more popular than~$M$. 
Note that $s_0$ obtains her top choice in~$M$, and $t_0$ can only be better off in~$M'$ if $s_0$ is worse off in~$M'$. 
Hence, as all remaining vertices in~$G'$ belong to a vertex or an edge gadget,  
there must exist a vertex or an edge gadget in which $M'$ beats~$M$, i.e., where more vertices prefer~$M'$ to~$M$ than vice versa. 
Suppose that $G'_i$ is such a gadget; the same argument works for edge gadgets. 
Note that the only vertices that do not (necessarily) obtain their top choice in $M$ are the endpoints of the edges $(s_i,a_{v_i})$ and $(b_{v_i},t_i)$, both in~$M$. Thus, either $s_i$ prefers $M'$ to $M$, or $t_i$ prefers $M'$ to~$M$, or both.
Furthermore, it also follows that at most four vertices in $G'_i$ can be better off in $M'$ when compared to~$M$, and so at most 
three
vertices of $G'_i$ may be better off in $M$ when compared to~$M'$. 
We distinguish between two cases: 
\begin{itemize} 
    \item Case A: both $s_i$ and $t_i$ prefer $M'$ to $M$. Then $M'(s_i)=a_x$ and $M'(t_i)=b_y$ for some $x$ and $y$ with $v_i \notin \{x,y\}$. 
    Since $a_{v_i} \prec_{s_i} a_x$, we get that $x$ precedes $v_i$ in $\pi$. By contrast, since $b_{v_i} \prec_{t_i} b_y$ we get that $y$ does \emph{not} precede $v_i$ in $\pi$. Hence, $x \neq y$.
    This implies that vertices $a_x$, $b_x$, $a_y$, and $b_y$ are four vertices that prefer $M$ to $M'$, a contradiction. 
    \item Case B: exactly one of $s_i$ and $t_i$ prefer $M'$ to~$M$, say $s_i$ (the case for $t_i$ is symmetric). Let $M'(s_i)=a_x$ for some $x \neq v_i$.
    If $M'(t_i)=M(t_i)$, then 
    the only vertices preferring $M'$ to $M$ can be $s_i$ and $a_{v_i}$, but both $a_x$ and $b_x$ prefer $M$ to $M'$, so $M'$ cannot beat $M$ in $G'_i$, a contradiction.  
    Otherwise, $M'(t_i) \neq M(t_i)$ and by our assumption, $t_i$ prefers $M$ to $M'$.
    Then $a_x$, $b_x$, and $t_i$ are three vertices in~$G'_i$ who prefer~$M$ to~$M'$, while at most three vertices in~$G'_i$ (namely, $s_i$, $a_{v_i}$, and $b_{v_i}$) may prefer~$M'$ to~$M$, a contradiction.
\end{itemize}
This proves that $M$ is indeed popular, as required, finishing our proof. 
\qed
\end{proof}

\repeattheorem{repthm_sociallystable}

\begin{proof}
We prove the theorem by slightly modifying the reduction from \myproblem{Multicolored Clique} given in the proof of Theorem~\ref{thm:masterlist-W1-hard}; we re-use all definitions from that proof. 
Recall the instance $I=(G',\prec,\omega,c,t,k=\binom{q}{2}+q)$ of \MUPMU{} constructed there. 
We define a modified instance $I'=(G',\prec,\omega,c',t,k'=0)$ by setting $c'(e)=1$ if $e$ is a consistency edge, and setting $c'(e)=0$ otherwise. 
Note that our modified reduction is a polynomial-time reduction.

We claim that there exists a feasible popular matching in~$I'$ if and only if the input graph~$G$ contains a clique of size~$q$. 

\paragraph{Direction ``$\Rightarrow$''.} Suppose that $M$ is a popular matching feasible for $I'$. The same arguments presented in the proof of Theorem~\ref{thm:masterlist-W1-hard} can be applied to define a clique in $G$, except for the argument to prove that $M$ cannot be blocked by any consistency edge---the only place where properties of the original cost function~$c$ and budget~$k$ were used. In the case of $I'$, however, this follows immediately from the fact that $c'(e)=1$ for any consistency edge~$e$ and our budget is $k'=0$. 

\paragraph{Direction ``$\Leftarrow$''.} Given a clique in~$G$, the corresponding matching~$M$ defined as in the proof of Theorem~\ref{thm:masterlist-W1-hard} has the property that $\bp(M) \cap F = \emptyset$, so the cost of $M$ with respect to the modified cost function~$c'$ is zero. Hence, $M$ is feasible in $I'$ as well.
\qed
\end{proof}

\repeattheorem{repthm_MLties}

\begin{proof}
Again, we are going to present a modification of the reduction given in the proof of Theorem~\ref{thm:k=1-bounded}, re-using all definitions there. One can observe that there is a simple reason why preferences in the  graph $\GG$ created in that reduction do not admit a master list (on either side): vertices  within a cycle $C^{c_j,\ell}$, $j \in [m]$ and $\ell \in [3]$, as well as within a cycle $\CC^{x_i,\lambda}$, $i \in [n]$ and $\lambda \in \{\texttt{t,f}\}$, are cyclic (in the sense that each vertex on such a cycle prefers the ``next'' vertex along the cycle to the ``previous'' vertex on the cycle, when traversing the cycle in one direction), and hence do not admit a master list on either side. We now present a modification that circumvents this problem with the help of ties, by simply replacing an edge in each of these cycles with a path of length~5 (i.e., making each of them longer by four edges). 

\paragraph{Construction.}
For each $j \in [m]$ and $\ell \in [3]$ we subdivide the edge $(a_1^{c_j,\ell},b_0^{c_j,\ell})$ of $\GG$ with newly introduced vertices $b_2^{c_j,\ell}$, $a_2^{c_j,\ell}$, $b_3^{c_j,\ell}$, and $a_3^{c_j,\ell}$, in this order. 
Similarly, for each $i \in [n]$ and $\lambda \in \{\TT,\FF\}$ we subdivide the edge $(b_3^{x_i,\lambda},a_0^{x_i,\lambda})$ with newly introduced vertices $a_4^{x_i,\lambda}$, $b_4^{x_i,\lambda}$, $a_5^{x_i,\lambda}$, and $b_5^{x_i,\lambda}$.
This way, for any cycle $C$ of the form~$\CC^{x_i,\lambda}$ or~$C^{c_j,\ell}$ in~$\GG$, we have created a corresponding cycle~$C^{\textup{ties}}$ in the obtained graph $\GG^{\textup{ties}}$. We will call such cycles \emph{base cycles} (both in~$\GG$ and in~$\GG^{\textup{ties}}$). 

We give the  preferences of the newly introduced vertices and their neighbors below (see Figure~\ref{fig:masterlist-ties}); all other preferences remain unchanged. Ties are denoted by angle brackets $\langle \cdot \rangle$.

\begin{longtable}{ll}
$a_1^{c_j,\ell}$: $\langle b_1^{c_j,\ell},b_2^{c_j,\ell} \rangle$ & where $j \in [m]$ and $\ell \in[3]$; \\
$b_2^{c_j,\ell}$: $\langle a_1^{c_j,\ell},a_2^{c_j,\ell} \rangle$ & where $j \in [m]$ and $\ell \in[3]$; \\
$a_2^{c_j,\ell}$: $b_2^{c_j,\ell},b_3^{c_j,\ell}$ & where $j \in [m]$ and $\ell \in[3]$; \\
$b_3^{c_j,\ell}$: $a_2^{c_j,\ell},a_3^{c_j,\ell}$ & where $j \in [m]$ and $\ell \in[3]$; \\
$a_3^{c_j,\ell}$: $\langle b_0^{c_j,\ell},b_3^{c_j,\ell} \rangle$ & where $j \in [m]$ and $\ell \in[3]$; \\
$b_0^{c_j,\ell}$: $\langle a_0^{c_j,\ell},a_3^{c_j,\ell} \rangle ,a_0^{c_j,\ell+1}$ & where $j \in [m]$ and $\ell \in[2]$;  \\
$b_0^{c_j,3}$: $\langle a_0^{c_j,3},a_3^{c_j,3} \rangle, a^{c_j}$ & where $j \in [m]$;  \\
$b_3^{x_i,\lambda}$: $\langle a_3^{x_i,\lambda},a_4^{x_i,\lambda} \rangle$ & where $i \in [n]$ and $\lambda \in \{\TT,\FF\}$; \\
$a_4^{x_i,\lambda}$: $\langle b_3^{x_i,\lambda},b_4^{x_i,\lambda} \rangle$ $\qquad \qquad$ & where $i \in [n]$ and $\lambda \in \{\TT,\FF\}$;  \\
$b_4^{x_i,\lambda}$: $a_4^{x_i,\lambda},a_5^{x_i,\lambda}$ & where $i \in [n]$ and $\lambda \in \{\TT,\FF\}$; \\
$a_5^{x_i,\lambda}$: $b_4^{x_i,\lambda},b_5^{x_i,\lambda}$ & where $i \in [n]$ and $\lambda \in \{\TT,\FF\}$; \\
$b_5^{x_i,\lambda}$: $\langle a_0^{x_i,\lambda},a_5^{x_i,\lambda} \rangle$ & where $i \in [n]$ and $\lambda \in \{\TT,\FF\}$; \\
$a_0^{x_i,\TT}$: $\langle b_0^{x_i,\TT},b_5^{x_i,\TT} \rangle, b_0^{x_i,\FF}$ & where $i \in [n]$; \\
$a_0^{x_i,\FF}$: $\langle b_0^{x_i,\FF},b_5^{x_i,\FF} \rangle, b^{x_i}$ & where $i \in [n]$.
\end{longtable}

\begin{figure}[th]
\makebox[\textwidth][c]{
\includegraphics[width=\textwidth]{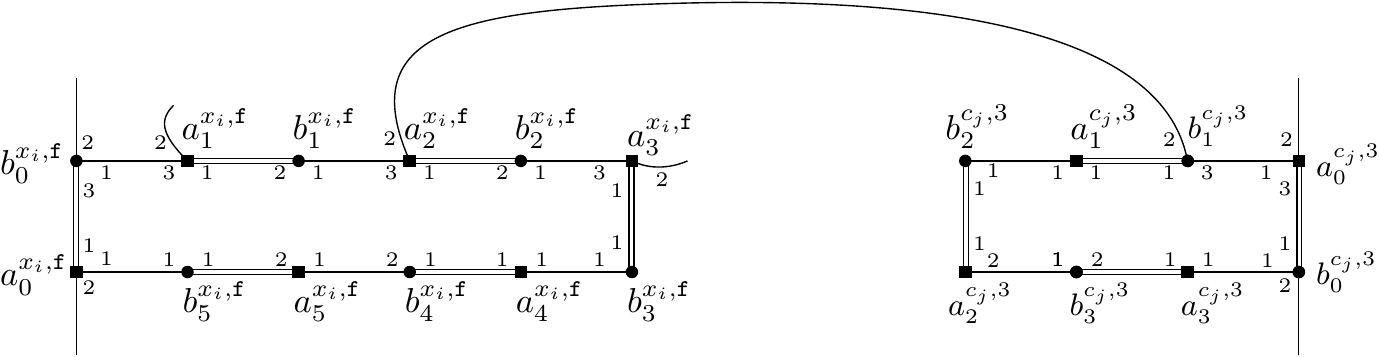}
}
\caption{Illustration of the reduction in the proof of Theorem~\ref{thm:masterlists-ties}, depicting how the graph $\GG$ is modified to obtain graph $\GG^{\textup{ties}}$. 
The figure assumes that the third literal in~$c_j$ is $x_i$ as a positive literal.
Edges leaving a given vertex~$a$ that are labelled with the same rank lead to vertices that are equally preferred by~$a$ and hence are contained in a tie in $a$'s preference list.
}
\label{fig:masterlist-ties}
\end{figure}

We keep all remaining parameters for both problems unchanged: $S=\{e^\star\}$ for \FUP{}, and for \MUPMU{} we set $c$, $k$, $\omega$ and $t$ as required in the statement of the theorem, with $f^\star=(u,v'_1)$ being the only edge with utility~1.

\paragraph{Further definitions.}
We extend the matching $\MM_0$  in the straightforward way into a matching~$\MM_0^{\textup{ties}}$:
for any $i \in [n]$ and $\lambda \in \{\TT,\FF\}$, we let $\MM_0^{\textup{ties}}(a_h^{x_i,\lambda})=b_h^{x_i,\lambda}$ for each \mbox{$h \in \{0, 1, \dots, 5\}$}, and similarly,
for each $j \in [m]$ and $\ell \in [3]$, we let \mbox{$\MM_0^{\textup{ties}}(a_h^{c_j,\ell})=b_h^{c_j,\ell}$} for each \mbox{$h \in \{0, 1, 2, 3\}$}; we set $\MM_0^{\textup{ties}}(a)=\MM_0(a)$ for each remaining vertex~$a$. 

Recall from the proof of Theorem~\ref{thm:k=1-bounded} that any popular matching $\MM$ in $\GG$ with 
$\bp(\MM)=\{e^\star\}$ has the following property: 
for any base cycle~$C$
either $\MM$ contains $C \cap M_0$, 
or $\MM$ contains $C \setminus M_0$.
Based on this fact, we can define a matching~$\MM^{\textup{ties}}$ in~$\GG^{\textup{ties}}$ corresponding to~$\MM$ as follows:
for each base cycle~$C$ in~$\GG$, 
if $M \supseteq \MM_0 \cap C$, then we add the edges $C^{\textup{ties}} \cap \MM_0^{\textup{ties}}$ to~$\MM^{\textup{ties}}$, otherwise we add the edges 
$C^{\textup{ties}} \setminus \MM_0^{\textup{ties}}$;
we further add all edges of~$\MM_0^{\textup{ties}}$ not incident to any base cycle to~$\MM^{\textup{ties}}$. 

\paragraph{Correctness.}
It is easy to see that if $\MM$ is a popular matching in $\GG$ with $\bp(\MM)=\{e^\star\}$, then $\MM^{\textup{ties}}$ is a popular matching in $\GG^{\textup{ties}}$ with $\bp(\MM^{\textup{ties}})=\{e^\star\}$. We are going to prove that the converse is true as well, i.e., any popular matching in $\GG^{\textup{ties}}$ blocked only by $e^\star$ is of the form $\MM^{\textup{ties}}$ for some popular matching~$\MM$ in~$\GG$ with $\bp(\MM)=\{e^\star\}$. 
Notice that this suffices to prove  that our modified reduction is correct (for both problems).

So let $M$ be a popular matching in $M$ blocked only by $e^\star$. It is clear that $(v,u'_1) \in M$. Using for $i=1, 2, \dots, n-1$ that the edge $(v_i,u'_{i+1})$ cannot block~$M$, we get that $(v_i,u'_{i+1}) \in M$ for each $i \in [n-1]$. 
In particular, no edge $(u'_i,b_0^{x_i,\TT})$ is in~$M$. 
This implies that $(a_0^{x_i,\TT},b_0^{x_i,\FF}) \notin M$:
indeed, $(a_0^{x_i,\TT},b_0^{x_i,\FF}) \in M$ would imply $(a_5^{x_i,\TT},b_5^{x_i,\TT}) \in M$ since $(a_0^{x_i,\TT},b_5^{x_i,\TT})$ cannot block $M$, 
and applying the same argument iteratively  it follows that 
$(a_h^{x_i,\TT},b_h^{x_i,\TT}) \in M$ for each $h \in [5]$. That would, however, leave $b_0^{x_i,\TT}$ unmatched in $M$ and hence $(b_0^{x_i,\TT},a_0^{x_i,\TT})$ would block~$M$, a contradiction. Therefore we know $(a_0^{x_i,\TT},b_0^{x_i,\FF}) \notin M$.
Using the same argument once more, we also get that $(a_0^{x_i,\FF},b^{x_i}) \notin M$. 

Note that analogous arguments can be used to show that $M$ contains $(u,v'_1)$ and all edges $(u_j,v'_{j+1})$ for $j \in [m-1]$, and that $M$ does not match any vertex of the form $b_0^{c_j,\ell}$ to its worst choice.
To prove our claim, it therefore remains to show that $M$ cannot contain any consistency edge.

\begin{figure}[th]
\makebox[\textwidth][c]{
\includegraphics[width=\textwidth]{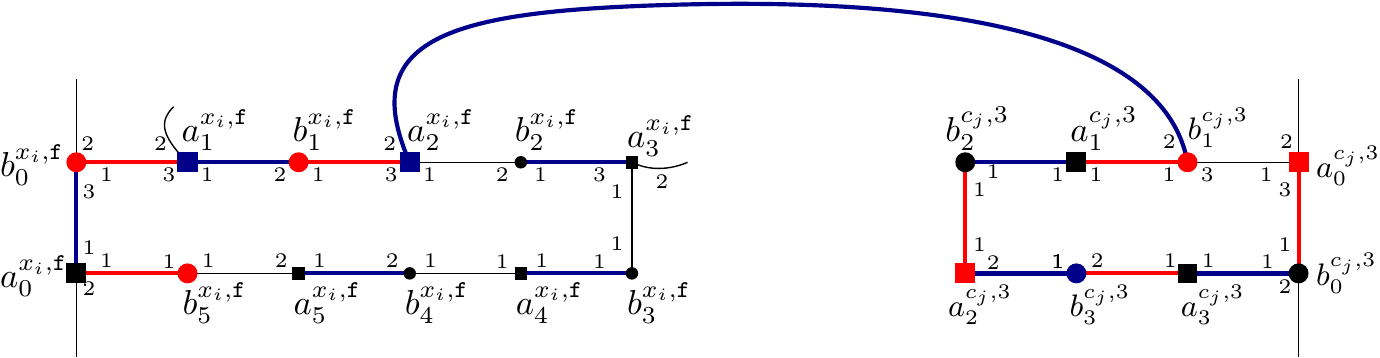}
}
\caption{Illustration of the matching $M'$ in the proof of Theorem~\ref{thm:masterlists-ties} in the case where the matching~$M$ (depicted in blue) contains the consistency edge $(a_2^{x_i,\FF},b_1^{c_j,3})$. 
The figure assumes $(a_0^{x_i,\FF},b_0^{x_i,\FF}) \in M$  but \mbox{$(a_0^{c_j,3},b_0^{c_j,3}) \notin M$}.
Edges of~$M' \setminus M$ are depicted in red.
Vertices in $V(M \triangle M')$ are emphasized by their somewhat greater size. Furthermore, vertices that prefer~$M'$ to~$M$ are shown in red, while those preferring~$M$ to~$M'$ are blue; indifferent vertices are black.
}
\label{fig:masterlist-ties-claim}
\end{figure}

Assume for the sake of contradiction that $M$ contains some consistency edge~$(a_h^{x_i,\lambda},b_1^{c_j,\ell,})$. 
We are going to define a matching $M'$ so that $M'$ is more popular than~$M$. 
Since $(a_h^{x_i,\lambda},b_h^{x_i,\lambda})$ cannot block~$M$, we know $(a_{h+1}^{x_i,\lambda},b_h^{x_i,\lambda}) \in M$, and applying the same argument iteratively we obtain that $(a_{h'+1}^{x_i,\lambda},b_{h'}^{x_i,\lambda}) \in M$ for each $h'=h,h+1, \dots, 3$.
We also get $(a_5^{x_i,\lambda},b_4^{x_i,\lambda}) \in M$, 
as otherwise
$(a_5^{x_i,\lambda},b_4^{x_i,\lambda})$ would block~$M$.
By contrast, if $h>1$, then since $(a_{h-1}^{x_i,\lambda},b_{h-1}^{x_i,\lambda})$ cannot block~$M$, it must be contained in~$M$, and again arguing iteratively we obtain that
$(a_{h'}^{x_i,\lambda},b_{h'}^{x_i,\lambda}) \in M$ for each $h' \in [h-1]$. Thus, $M$ assigns either~$b_0^{x_i,\lambda}$ or~$b_5^{x_i,\lambda}$ to~$a_0^{x_i,\lambda}$, 
leaving the other vertex unmatched. 
Now, we define $M'$ to contain the edges~$\{(a_h^{x_i,\lambda},b_h^{x_i,\lambda}):h \in \{0,1,\dots,5\}\}$
if $(a_0^{x_i,\lambda},b_0^{x_i,\lambda}) \notin M$ and we define $M'$
to contain the edges~$\{(a_{h+1 \! \! \mod 6}^{x_i,\lambda},b_h^{x_i,\lambda}):h \in \{0,1,\dots,5\}\}$
otherwise. We use the same reasoning for the cycle $C^{c_j,\ell}$ and define $M'$ analogously on the vertices of~$C^{c_j,\ell}$. For each remaining vertex~$a$ we let $M'(a)=M(a)$. 
It is not hard to see by simply checking all vertices in~$V(M \triangle M')$ that there are at least two more vertices preferring $M'$ to $M$ than vice versa (see  Figure~\ref{fig:masterlist-ties-claim}), showing that $M'$ is indeed more popular than $M$, a contradiction.

\paragraph{Master lists.}
To finish the proof of our theorem, it remains to show that the reduction has the promised properties, in particular that preferences on both sides admit a master list. 
To this end, for $h=1,2,3$ we define $I_{h}$ as the set of those pairs $(c_j,\ell)$ where the $\ell$-th literal in clause $c_j$ is the $h$-th occurrence of some variable in the input formula.
We define the two master lists as follows.

\begin{longtable}{ll}
$\mathcal{L}_A$: & 
$u, \big\{ \big\langle a_1^{c_j,\ell}, a_2^{c_j,\ell} \big\rangle \big\}_{j \in [m],\ell \in [3]},
\big\{ 
\big\langle a_3^{x_i,\lambda}, a_4^{x_i,\lambda}  
\big\rangle 
\big\}_{i \in [n],\lambda \in \{\TT,\FF \}}$,\\[10pt]
& 
$
\big\{a_2^{x_i,\lambda}\}_{i \in [n],\lambda \in \{\TT,\FF \}} ,
 \{a_1^{x_i,\lambda}\}_{i \in [n],\lambda \in \{\TT,\FF \}} ,u'_1,u'_2,\dots, u'_n,
$ \\[6pt]
& $\big\langle a_0^{c_1,1}, a_3^{c_1,1} \big\rangle, u_1, 
\big\langle a_0^{c_2,1}, a_3^{c_2,1} \big\rangle, u_2, \dots, 
u_{m-1},
\big\langle a_0^{c_m,1}, a_3^{c_m,1} \big\rangle,$ \\[6pt]
&
$
\big\{ \big\langle a_0^{c_j,2},a_3^{c_j,2} \big\rangle 
\big\}_{j \in [m]},
\big\{ \big\langle a_0^{c_j,3},a_3^{c_j,3} \big\rangle \big\}_{j \in [m]},
 \big\{a^{c_j}\big\}_{j \in [m]},
$ \\[6pt]
& $
\big\{ \big\langle a_0^{x_i,\TT},a_5^{x_i,\TT} \big\rangle \big\}_{i \in [n]},
\big\{ \big\langle a_0^{x_i,\FF},a_5^{x_i,\FF} \big\rangle \big\}_{i \in [n]}.$ 
\end{longtable}


\begin{longtable}{ll}
$\mathcal{L}_B$:& $v,
\big\{ \big\langle  b_3^{x_i,\lambda}, b_4^{x_i,\lambda} \big\rangle \big\}_{i \in [n],\lambda \in \{\TT,\FF \} },
\big\{ \big\langle b_1^{c_j,\ell}, b_2^{c_j,\ell} \big\rangle \big\}_{(c_j,\ell) \in I_3 }, 
\big\{b_2^{x_i,\lambda} \big\}_{i \in [n],\lambda \in \{\TT,\FF\}}$, \\[6pt]
& 
$\big\{ \big\langle  b_1^{c_j,\ell}, b_2^{c_j,\ell} \big\rangle \big\}_{(c_j,\ell) \in I_2 }, 
\big\{b_1^{x_i,\lambda} \big\}_{i \in [n],\lambda \in \{\TT,\FF\}},
\big\{ \big\langle  b_1^{c_j,\ell}, b_2^{c_j,\ell} \big\rangle \big\}_{(c_j,\ell) \in I_1 }, 
$
\\[6pt]
& $
\big\langle b_0^{x_1,\TT},b_5^{x_1,\TT} \big\rangle,v_1,
\big\langle b_0^{x_2,\TT},b_5^{x_2,\TT} \big\rangle,v_2,
\dots, 
v_{n-1},
\big\langle b_0^{x_n,\TT},b_5^{x_n,\TT} \big\rangle, $
\\[6pt] 
& $
\big\{ \big\langle b_0^{x_i,\FF},b_5^{x_i,\FF} \big\rangle \big\}_{i \in [n]},
\big\{b^{x_i}\big\}_{i \in [n]} , v'_1,v'_2, \dots, v'_m,$ 
\\[6pt]
&
$
\big\{ \big\langle b_0^{c_j,1},b_3^{c_j,1} \big\rangle \big\}_{j \in [m]},
\big\{ \big\langle b_0^{c_j,2},b_3^{c_j,2} \big\rangle \big\}_{j \in [m]},
\big\{ \big\langle b_0^{c_j,3},b_3^{c_j,3} \big\rangle \big\}_{j \in [m]}.
$ \\[6pt]
\end{longtable}

It is not hard (though tedious) to verify that  preferences in the constructed instance indeed admit the master lists given above. 
\qed
\end{proof}

\subsection{Missing proofs and details from Section~\ref{sec:pareto}}
\label{sec:app-pareto}

We start with an example showing that  a Pareto-optimal matching~$M$ in a preference system~$(G,\prec)$ may not be Pareto-optimal in $(G,\prec)-\bp_G(M)$.

\begin{myexample}
\label{example:PO}
%
Instead of a formal definition, we set the structure and preferences of our example as depicted in Figure~\ref{fig:example-PO}. Assume $k=2$ and that we are looking for a matching that is blocked by exactly the edges contained in $S=\{(a_1,b_2),(a_2,b_1)\}$. Assume that edges in $S$ have utility~1, while all other edges have utility~2. 
There are two stable matchings in $G-S$, both of them having the same utility.
Observe that the matching $M_a=\{(a_i,b_i):i \in [3]\}$ highlighted in part~(a) of Figure~\ref{fig:example-PO} 
is Pareto-optimal in~$G-S$; 
however, it is not Pareto-optimal in~$G$: switching edges along the cycle $(a_1,b_1,a_2,b_2)$ yields a 
matching that is a Pareto-improvement for~$M_a$, but has less utility than~$M_a$. 
By contrast, the matching $M_b=\{(a_1,b_1),(a_2,b_3),(a_3,b_2)\}$ highlighted in part~(b) of Figure~\ref{fig:example-PO} 
is Pareto-optimal in $G-S$ as well as in $G$. 
\end{myexample}

\begin{figure}[th]
\makebox[\textwidth][c]{
\includegraphics[scale=0.8]{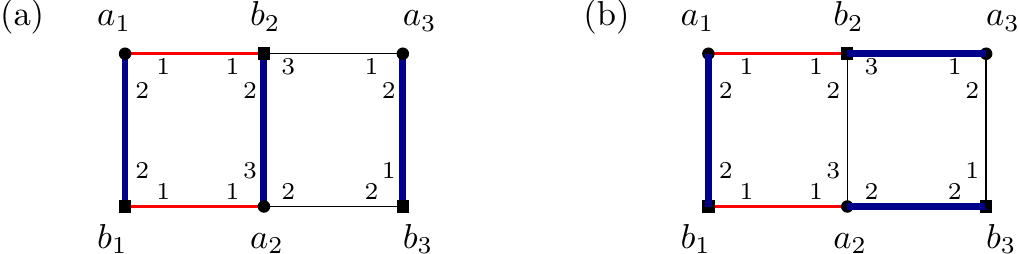}
}
\caption{An example illustrating why the simplest approach for finding a Pareto-optimal matching with at most $k$ blocking edges and maximal utility does not work. 
}
\label{fig:example-PO}
\end{figure}

\repeattheorem{repthm_PO_W1hard}
\begin{proof}
Recall the reduction presented in the proof Theorem~\ref{thm:masterlist-W1-hard}: given an instance~$(G=(V,E),q)$ of
 \myproblem{Multicolored Clique} with $V$ partitioned into $q$ sets $V_1, \dots, V_q$, we 
constructed an equivalent instance $I=(G',\prec,\omega,c,t,k=q+\binom{q}{2})$ of \MUPMU{}. 
We claim that this construction is correct here as well
in the sense that $(G',\prec)$ admits a matching~$M$ of size $|M_s|+1=|V(G')|/2$ and with $\bp_{G'}(M)\leq k$ if and only if $G$ has a clique $v_1, \dots, v_q$ in $G$ with $v_i \in V_i$ for each $i \in [q]$. 

On the one hand, given such a clique in $G$, the corresponding matching $M$ defined in the proof of Theorem~\ref{thm:masterlist-W1-hard} 
is clearly a matching that is larger than the unique stable matching~$M_s$ in $(G',\prec)$ and has exactly~$k$ blocking pairs.

On the other hand, let $M$ be a matching in $(G',\prec)$ that is larger than $M_s$ and has at most~$k$ blocking edges. Note that $M$
is feasible in instance $I$ of \MUPMU{}, but we do not know whether it is popular in $(G',\prec)$ or not.
Nevertheless, observe that in the proof of Theorem~\ref{thm:masterlist-W1-hard} we only use the popularity of~$M$
in order to show that certain vertices must be matched by~$M$. This is, however, obvious now, since $M$ is a complete matching. 
Therefore,  we can apply the same arguments we used in the proof of Theorem~\ref{thm:masterlist-W1-hard}, 
making use of the completeness of~$M$ instead of its popularity, to obtain a clique in~$G$ as required. 

Thus, the reduction is indeed correct and proves the theorem.
\qed
\end{proof}

\repeatcorollary{repcor_PO}

\begin{proof}
The corollary is an immediate consequence of Theorem~\ref{thm:t1-instability} and the following observation: $G$ contains a matching of size~$t$ with at most~$k$ blocking edges if and only if $G$ contains a Pareto-optimal matching of size~$t$ with at most~$k$ blocking edges. To see the non-trivial direction of this observation, consider any matching $M$ of size~$t$ with $|\bp(M)|\leq k$. By applying Pareto-improvements to~$M$ as long as possible, we obtain a Pareto-optimal matching $M'$ where no vertex is worse off than in~$M$. 
This implies that 
an edge that blocks $M'$ also blocks $M$ 
(note that no edge in $M \setminus M'$ blocks $M'$). 
It follows that $|\bp_G(M')|\leq k$, and
it is also clear that $|M'|\geq |M| =t$, proving the observation.
\qed
\end{proof}

\end{document}